\documentclass[11pt]{article} 
\usepackage{ifthen}
\usepackage{xspace}
\usepackage{float}
\usepackage{subcaption}
\usepackage{amsmath}
\usepackage{amsthm}
\usepackage{amstext}
\usepackage{amssymb}
\usepackage{amsfonts}
\usepackage{bm}
\usepackage{algorithm}
\usepackage[noend]{algpseudocode} 
\usepackage{supertech}
\usepackage[margin=1in]{geometry}
\usepackage[numbers,sort&compress,square]{natbib}
\usepackage{mathdots}
\usepackage{textlater}
\usepackage{relate}
\usepackage{times}
\usepackage{thmtools}
\usepackage{thm-restate}
\usepackage{tikz}
\usetikzlibrary{calc, fit, shapes,intersections,arrows,matrix, topaths}
\usetikzlibrary{decorations,decorations.pathreplacing,patterns, shadows}

\newif\ifbuildfigure \buildfiguretrue  

\definecolor{magenta4}{rgb}{0.5625,0,0.5625}
\definecolor{green4}{rgb}{0,0.5625,0}
\definecolor{orange4}{rgb}{0.98,0.31,0.09}

\newif\ifjournal 
 \journaltrue

\newif\ifnotes
\notestrue
\ifnotes
\newcommand{\sam}[1]{{\ifnotes \scriptsize \textcolor{green4}{Sam: {#1}} \fi}}
\newcommand{\rob}[1]{{\ifnotes \scriptsize \textcolor{purple}{Rob: {#1}} \fi}}
\newcommand{\shikha}[1]{\ifnotes {\noindent \scriptsize  \textcolor{blue} {Shikha: {#1}}} \fi{}}

\newcommand{\mab}[1]{{\ifnotes \scriptsize \textcolor{red}{Michael:
      {#1}} \fi}}
\newcommand{\mayank}[1]{{\ifnotes \scriptsize \textcolor{yellow}{Mayank: {#1}} \fi}}

\else
\newcommand{\shikha}[1]{}
\newcommand{\mayank}[1]{}
\newcommand{\martin}[1]{}
\newcommand{\sam}[1]{}
\newcommand{\mab}[1]{}
\newcommand{\rob}[1]{}
\fi

\newcommand{\build}[1][]{\mbox{\textsc{Build}}\ifthenelse{\equal{#1}{}}{}{\ensuremath{(}#1\ensuremath{)}}\xspace}
\newcommand{\lookup}[1][]{\mbox{\textsc{Lookup}}\ifthenelse{\equal{#1}{}}{}{\ensuremath{(}#1\ensuremath{)}}\xspace}
\newcommand{\update}[1][]{\mbox{\textsc{Update}}\ifthenelse{\equal{#1}{}}{}{\ensuremath{(}#1\ensuremath{)}}\xspace}
\newcommand{\delete}[1][]{\mbox{\textsc{Delete}}\ifthenelse{\equal{#1}{}}{}{\ensuremath{(}#1\ensuremath{)}}\xspace}

\newcommand{\amqbuild}[1][]{\mbox{\textsc{Init}}\ifthenelse{\equal{#1}{}}{}{\ensuremath{(}#1\ensuremath{)}}\xspace}
\newcommand{\amqlookup}[1][]{\mbox{\textsc{Lookup}}\ifthenelse{\equal{#1}{}}{}{\ensuremath{(}#1\ensuremath{)}}\xspace}
\newcommand{\amqupdate}[1][]{\mbox{\textsc{Adapt}}\ifthenelse{\equal{#1}{}}{}{\ensuremath{(}#1\ensuremath{)}}\xspace}
\newcommand{\amqinsert}[1][]{\mbox{\textsc{Insert}}\ifthenelse{\equal{#1}{}}{}{\ensuremath{(}#1\ensuremath{)}}\xspace}
\newcommand{\amqdelete}[1][]{\mbox{\textsc{Delete}}\ifthenelse{\equal{#1}{}}{}{\ensuremath{(}#1\ensuremath{)}}\xspace}

\newcommand{\dictbuild}[1][]{\mbox{\textsc{Dict-Build}}\ifthenelse{\equal{#1}{}}{}{\ensuremath{(}#1\ensuremath{)}}\xspace}
\newcommand{\dictlookup}[1][]{\mbox{\textsc{Dict-Lookup}}\ifthenelse{\equal{#1}{}}{}{\ensuremath{(}#1\ensuremath{)}}\xspace}
\newcommand{\dictupdate}[1][]{\mbox{\textsc{Dict-Update}}\ifthenelse{\equal{#1}{}}{}{\ensuremath{(}#1\ensuremath{)}}\xspace}
\newcommand{\dictinsert}[1][]{\mbox{\textsc{Dict-Insert}}\ifthenelse{\equal{#1}{}}{}{\ensuremath{(}#1\ensuremath{)}}\xspace}
\newcommand{\dictdelete}[1][]{\mbox{\textsc{Dict-Delete}}\ifthenelse{\equal{#1}{}}{}{\ensuremath{(}#1\ensuremath{)}}\xspace}

\newcommand{\aadbuild}[1][]{\mbox{\textsc{Init}}\ifthenelse{\equal{#1}{}}{}{\ensuremath{(}#1\ensuremath{)}}\xspace}
\newcommand{\aadlookup}[1][]{\mbox{\textsc{Lookup}}\ifthenelse{\equal{#1}{}}{}{\ensuremath{(}#1\ensuremath{)}}\xspace}
\newcommand{\aadinsert}[1][]{\mbox{\textsc{Insert}}\ifthenelse{\equal{#1}{}}{}{\ensuremath{(}#1\ensuremath{)}}\xspace}
\newcommand{\aaddelete}[1][]{\mbox{\textsc{Delete}}\ifthenelse{\equal{#1}{}}{}{\ensuremath{(}#1\ensuremath{)}}\xspace}

\newcommand{\evolve}[1][]{\mbox{\textsc{Evolve}}\ifthenelse{\equal{#1}{}}{}{\ensuremath{(}#1\ensuremath{)}}\xspace}

\newcommand{\adaptivityGame}[1][]{\mbox{\textsc{adaptivity-game}}\ifthenelse{\equal{#1}{}}{}{\ensuremath{(}#1\ensuremath{)}}\xspace}

\newcommand{\forced}[1][]{\Phi\xspace}
\newcommand{\tfp}[1][]{\mbox{\textsc{TFP}}\ifthenelse{\equal{#1}{}}{}{\ensuremath{(}#1\ensuremath{)}}\xspace}

\newcommand{\extend}[1][]{\mbox{\textsc{Adapt}}\ifthenelse{\equal{#1}{}}{}{\ensuremath{(}#1\ensuremath{)}}\xspace}
\newcommand{\absfp}[1][]{{absolute false positive}\xspace}
\newcommand{\absset}{{\textnormal{AFP}}}
\newcommand{\oabsset}{{\textnormal{OFP}}}
\newcommand{\absfps}{{absolute false positives}\xspace}
\newcommand{\oabsfp}{{original absolute false positive}\xspace}
\newcommand{\oabsfps}{{original absolute false positives}\xspace}
\newcommand{\fix}[1][]{\textnormal{FIX}}

\newcommand{\prefix}{\sqsubseteq}
\newcommand{\FP}{\text{FP}}

\newcommand{\fprate}{sustained false-positive rate\xspace}

\newcommand{\fpprob}{false-positive probability\xspace}

\newenvironment{proofsketch}{%
  \proof}{\endproof}

\let\oldtexttt\texttt
\renewcommand{\texttt}[1]{\xspace{\normalfont{\oldtexttt{#1}}}\xspace}

\renewcommand{\epsilon}{\varepsilon}

\newcommand{\E}{\mathbf{E}}

\renewcommand{\epsilon}{\varepsilon}

\newcommand{\calA}{{\cal A}}
\newcommand{\calS}{{\cal S}}
\newcommand{\calfS}{{\cal S^*}}

\newcommand{\calU}{{\cal U}}

\newcommand{\calP}{{\cal P}}

\newcommand{\calQ}{{\cal Q}}
\newcommand{\calZ}{{\cal Z}}
\newcommand{\calO}{\ensuremath{{\cal O}}\xspace}

\newcommand{\poly}{\textnormal{poly}}

\newif\ifproofof
\proofoftrue

\ifproofof
\newenvironment{proofof}[1]{$ $\newline \noindent{\em Proof of {#1}. }\ignorespaces}{\putqed \newline}

\else

\fi

\newcommand{\pparagraph}[1]{\vspace{0.07in}\noindent{\emph{#1}}}  

\newcommand{\baseline}{baseline\xspace}

 \newcommand{\mextend}{\textsc{Extend}}

 \newcommand{\arev}[1]{$\mbox{\textsc{RevLookup}}_{\calP}(#1)$\xspace}
 \newcommand{\rev}{$\mbox{\textsc{RevLookup}}_{\calP}$\xspace}

\newcommand{\dict}{{\mathbf D}}

\newcommand{\rem}{{\mathbf R}} 
\newcommand{\loc}{{\mathbf L}} 

\newcommand{\lmem}{\loc}

\newcommand{\present}{\textsc{Present}\xspace}
\newcommand{\absent}{\textsc{Absent}\xspace}

\newcommand{\True}{\textsc{True}\xspace}

\newcommand{\acksmall}{}

\author{
Michael A.~Bender\thanks{
Stony Brook University, 
Stony Brook, NY 11794-4400, USA. 
Email:~\texttt{bender@cs.stonybrook.edu}.}
\and
Martin Farach-Colton\thanks{
Rutgers University, Piscataway NJ 08855, USA.
Email:~\texttt{farach@cs.rutgers.edu}.}
\and
Mayank Goswami\thanks{
Queens College, CUNY, New York, USA. 
Email:~\texttt{mayank.goswami@qc.cuny.edu}.}
\and
Rob Johnson\thanks{
VMware Research,
Creekside F,
3425 Hillview Ave,
Palo Alto, CA 94304.
Email:~\texttt{robj@vmware.com}.}
\and
Samuel McCauley\thanks{
Wellesley College, Wellesley, MA 02481 USA. 
Email:~\texttt{\{smccaule, shikha.singh\}@wellesley.edu}.}
\and
Shikha Singh\footnotemark[5]  
}

\newcommand{\realtitle}{Bloom Filters, Adaptivity, \\and the Dictionary Problem}

\date{}
\title{\realtitle \footnote{
This research was supported in part by NSF grants
CCF 1114809,
CCF 1217708,
CCF 1218188,
CCF 1314633,
 CCF 1637458, 
IIS 1247726, 
IIS 1251137, 
CNS 1408695, 
CNS 1408782, 
CCF 1439084, 
CCF-BSF 1716252, 
CCF 1617618, 
IIS 1541613,
and 
CAREER Award CCF 1553385, 
as well as NIH grant 1U01CA198952-01,  
by the European Research Council under the European Union's 7th Framework Programme (FP7/2007-2013)~/~ERC grant agreement no. 614331,
by Sandia National Laboratories, EMC, Inc, and  NetAPP, Inc.
BARC, Basic Algorithms Research Copenhagen, is supported by the VILLUM Foundation grant 16582. 
}
}

\begin{document} 
\maketitle

\begin{abstract}

  An {approximate membership query}
  \boldmath
  data structure ({AMQ})---such as a Bloom, quotient, or cuckoo
  filter---maintains a compact, probabilistic representation of a set
  $\calS$ of keys from a universe~$\calU$. It supports lookups and
  inserts.  Some AMQs also support deletes.
  A query for $x\in\calS$ returns \present.  A query for 
  $x\not\in\calS$ returns \present with a tunable
  \defn{\fpprob}~$\epsilon$, and otherwise returns \absent.

  AMQs are widely 
used to speed up dictionaries that are stored
  remotely (e.g., on disk or across a network).  The AMQ is stored
  locally (e.g., in memory).   The remote dictionary is only 
accessed
  when the AMQ returns \present.
  Thus, the primary performance metric of an AMQ is how often it
  returns \absent for negative queries.

  Existing AMQs offer weak guarantees on the number of false positives
  in a sequence of queries.  The false-positive
  probability $\epsilon$ holds only for a single query.  It is easy
  for an adversary to drive an AMQ's false-positive rate towards 1 by
  simply repeating false positives.

  This paper shows what it takes to get strong guarantees on the
  number of false positives.  We say that an AMQ is \defn{adaptive} if
  it guarantees a false-positive probability of $\epsilon$ for every
  query, \emph{regardless of answers to previous queries}.


We establish upper and lower bounds for adaptive AMQs.  Our lower
bound shows that it is impossible to build a small adaptive AMQ, even
when the AMQ is immediately told whenever a query is a false positive.
On the other hand, we show that it is possible 
to
maintain an AMQ that uses the same amount of local space as a
non-adaptive AMQ (up to lower order terms), performs all queries and
updates in constant time, and guarantees that each negative query
to the dictionary accesses remote storage with probability $\epsilon$,
independent of the results of past queries.  Thus, we show that
adaptivity can be achieved effectively for free.

\end{abstract}


\sloppy

\section{{\bf I{\small NTRODUCTION}}}
\label{sec:intro}  

  An {approximate membership query}
  data structure ({AMQ})---such as a Bloom~\cite{Bloom70,BroderMi04}, quotient~\cite{BenderFaJo12,PandeyBeJo17}, single hash~\cite{PaghPaRa05}, or cuckoo~\cite{FanAnKa14} filter---maintains a compact, probabilistic representation of a set
  $\calS$ of keys from a universe~$\calU$. 
  It supports lookups and
  inserts.  Some AMQs also support deletes.
  A positive query for $x\in\calS$ returns \present.  A negative
  query for $x\not\in\calS$ returns \present with a tunable
  \defn{\fpprob}~$\epsilon$, and otherwise returns \absent.

AMQs are used because they are small.
An optimal AMQ can encode
a set $\calS \subseteq \calU$, where $|\calS|=n$ and $|\calU|=u$, with
a false-positive probability $\epsilon$ using
$\Theta(n \log (1/\epsilon))$ bits~\cite{CarterFlGi78}. In contrast,
an error-free representation 
of $\calS$ takes
$\Omega(n\log u)$ bits.

One of the main uses of AMQs is
to speed up dictionaries~\cite{ZhuJiWa04, 
  DengRa06, 
  EppsteinGoMi2017, 
  CohenMa03, 
  TarkomaRoLa12,%
  farachcoltonfm2009,%
  BroderMi04}.  
Often, there is not enough local storage
(e.g., RAM) to store the dictionary's internal state, $\dict$. 
Thus, $\dict$ must be maintained remotely
(e.g., on-disk or across a network), and accesses to $\dict$ are
expensive.  By maintaining a 
local AMQ for the set $\calS$ of keys occurring in $\dict$, the dictionary can avoid accessing $\dict$ on
most negative queries: if the AMQ says that a key is not in $\calS$,
then no query to $\dict$ is necessary.  

Thus, the primary performance metric of an AMQ is how well it enables a
dictionary to avoid these expensive accesses to $\dict$.  The fewer false positives an
AMQ returns on a sequence of queries, the more effective it is.

\pparagraph{AMQ guarantees.}  
Existing AMQs offer weak guarantees on the number of false positives they will
return for a sequence of queries.  The
false-positive probability of $\epsilon$ holds only for a single query.
It does not  extend to multiple queries, because queries can be
correlated.  It is easy for an adversary to drive
an AMQ's false-positive rate towards 1 by simply repeating false-positives.

Even when the adversary is oblivious, i.e., it 
 selects $n$ queries without regard to the results of
previous queries, existing AMQs have weak guarantees.
With probability $\epsilon$, a random query is a false positive, and
repeating it $n$ times results in a false-positive rate of 1. Thus, 
even when the adversary is oblivious, existing AMQs can have $O(\epsilon n)$
false positives in expectation but not with high probability. 
This distinction has implications: 
Mitzenmacher et al.~\cite{MitzenmacherPo17} show that on network traces, 
existing AMQs are suboptimal because they do not adapt to false
positives. 


\pparagraph{Adaptive AMQs.}  We define an \defn{adaptive AMQ} to be an
AMQ that returns \present with probability at most $\epsilon$ for
every negative query, \emph{regardless of answers to previous queries}.
For a dictionary using an adaptive AMQ, \emph{any} sequence of $n$
negative queries will result in $O(\epsilon n)$ false positives,
with high probability.  This gives a strong bound on the number of (expensive) negative
accesses that the dictionary will need to make to $\dict$.  This is true
even if the queries are selected by an adaptive adversary.


Several attempts have been made to move towards adaptivity (and beyond oblivious adversaries).
Naor and Yogev~\cite{NaorYo15} considered an adaptive adversary that
tries to increase the false-positive rate by discovering
collisions in the AMQ's hash functions, but they explicitly forbade
the adversary from repeating queries. 
Chazelle et al.~\cite{ChazelleKiRu04}
introduced bloomier filters, which can be updated to specify a white
list, which are elements in $\calU-\calS$ on which the AMQ may not
answer \present.  However, bloomier filters are space efficient only
when the white list is specified in advance, which makes them unsuitable
for adaptivity. 
Mitzenmacher et al.~\cite{MitzenmacherPo17}
proposed an elegant variant of the cuckoo filter that stores part of
the AMQ locally and part of it remotely in order to try to achieve
adaptivity.  They empirically show that their data structure helps
maintain a low false-positive rate against queries that have temporal
correlation.

However, no existing AMQ is provably adaptive.

\pparagraph{Feedback, local AMQs, and remote representations.}  When
an AMQ is used to speed up a dictionary, the dictionary always detects which are the AMQ's false positives and which are the true positives. 
%
Thus, the dictionary can provide this feedback to the AMQ. This feedback is free because 
it does not require any additional accesses to $\dict$ beyond what was used to answer the query.
%


In this paper we show that, even with this feedback, it is impossible to construct an
adaptive AMQ that uses less than $\Omega(\min\{n\log\log u,n\log n\})$ bits of space;
see \thmref{lowerfinal}.
That is, even if an AMQ is told  which are the true and false
positives, adaptivity requires 
large space.


This lower bound would appear
to kill the whole idea of adaptive AMQs, since one of the key
ideas of an AMQ is to be small enough to fit in local storage. 
Remarkably, 
%
%
 efficient adaptivity is still achievable.

The way around this impasse is to partition an AMQ's state into a small local state $\loc$ and a larger remote state $\rem$.
The AMQ can still have good
performance, provided it access the remote state 
infrequently.


We show how to make an adaptive AMQ that consumes no more local space than the best non-adaptive AMQ (and much less than a Bloom filter).   We call this data structure a \defn{broom filter} (because it cleans up its mistakes).
%
The broom filter accesses $\rem$ only when the AMQ receives feedback that
it returned a false positive.  

When used to filter accesses to a remote dictionary $\dict$, the AMQ's
accesses to $\rem$ are ``free''---i.e. they do not asymptotically
increase the number of accesses to remote storage----because the AMQ
access $\rem$ only when the dictionary accesses $\dict$.

Our lower bound shows that partitioning is essential to creating a
space-efficient adaptive AMQ. Indeed, the adaptive cuckoo filter of
Mitzenmacher et al.~\cite{MitzenmacherPo17} 
also partitions its state into local and remote components, but it does not have
the strong theoretical adaptivity guarantees of the broom filter.

The local component, $\loc$, of the broom filter is itself a non-adaptive AMQ
plus $O(n)$ bits for adaptivity.  The purpose of $\rem$ is to
provide a little more information to help $\loc$ adapt.   

Thus, we have a dual view of adaptivity that helps us interpret the
upper and lower bounds.  The local representation $\loc$ is an AMQ in
its own right.  The remote representation $\rem$ is an ``oracle'' that
gives extra feedback to $\loc$ whenever there is a false positive.
Because $\rem$ is simply an oracle, all the heavy lifting is in the
design of $\loc$. In the broom filter, $\rem$ enables $\loc$ to
identify an element $y\in\calS$ that triggered the false positive. 

%
%

Putting these results together, we pinpoint how much information is
needed for an adaptive AMQ to update its local information.  The lower
bound shows that simply learning if the query is a false positive is
not sufficient. But if this local information is augmented with
asymptotically free remote lookups, then adaptivity is achievable.

\pparagraph{A note on optimality.}  The broom filter dominates
existing AMQs in all regards.  Its local state by itself is an
optimal conventional AMQ: it uses optimal space up to lower-order
terms, and supports queries and updates in constant time with high probability.  Thus the remote state is only for adaptivity.
For comparison, a Bloom filter has a lookup time of
$O(\log\frac{1}{\epsilon})$, the space is suboptimal, and the filter
does not support deletes.  More recent AMQs~\cite{PaghPaRa05, BenderFaJo12, FanAnKa14, PandeyBeJo17} also fail to match the
broom filter on one or more of these criteria, even leaving aside
adaptivity.  Thus, we show that adaptivity has no cost.



\newcommand{\obj}{}

\algblockdefx{Method}{EndMethod}%
             [2]{\textbf{method} \textsc{#1}(#2)}%
             {}
\algnotext{EndMethod}

\algblockdefx{game}{Endgame}%
             [2]{\textbf{game} \textsc{#1}(#2)}%
             {}
\algnotext{Endgame}
             
\begin{figure*}[t]
  \centering
  \hspace{-.275in}
  \begin{subfigure}[t]{2.39in}
    \begin{algorithmic}
      \game{adaptivity-game}{$\calA,n,\epsilon$}
        \State $\calO \gets$ \Call{setup}{$n$, $\epsilon$}
        \State $x'\gets \calA^\calO(n,\epsilon)$
        \State $b\gets\calO.\amqlookup[x']$
        \State\Return $(b=\present) \wedge (x'\not\in\calO.\calS)$
      \Endgame
    \end{algorithmic}
  \end{subfigure}
  \hspace{-.2in}
  \begin{subfigure}[t]{2.07in}
    \begin{algorithmic}
      \Function{setup}{$n$,$\epsilon$}
        \State $\calO.\rho\stackrel{\$}{\gets}\{0,1\}^{\mathbb{N}}$
        \State $\calO.\calS\gets\emptyset$
        \State $(\calO.\loc,\calO.\rem)\gets$\amqbuild[$n,\epsilon,\calO.\rho$]
        \State \Return \calO
      \EndFunction
    \end{algorithmic}
  \end{subfigure}
  \hspace{-.1in}
  \begin{subfigure}[t]{2.3in}
    \begin{algorithmic}
      \Method{\calO.\amqlookup}{$x$}
        \State $(\obj\loc,b)\gets$\amqlookup[$\obj\loc,x,\obj\rho$]
        \If{$(b=$\present) $\wedge (x\not\in\obj{}\calS)$}
          \State $(\obj\loc,\obj\rem)\gets$\amqupdate[$(\obj\loc,\obj\rem),x,\obj\rho$]
          \EndIf
        \State\Return $b$
      \EndMethod
    \end{algorithmic}
  \end{subfigure}  \vspace{.13in}\\
  \begin{subfigure}[t]{2.4in}
    \begin{algorithmic}
      \Method{\calO.\amqinsert}{$x$}
        \If{$|\calS| < n \wedge x\not\in \calS$}
          \State $(\obj\loc,\obj\rem)\gets$\amqinsert[$(\obj\loc,\obj\rem),x,\obj\rho$]
          \State $\obj{}\calS\gets\obj{}\calS \cup \{x\}$
        \EndIf
      \EndMethod
    \end{algorithmic}
  \end{subfigure}
  \hspace{.1in}
  \begin{subfigure}[t]{2.4in}
    \begin{algorithmic}
      \Method{\calO.\amqdelete}{$x$}
        \If{$x\in \calS$}
          \State $(\obj\loc,\obj\rem)\gets$\amqdelete[$(\obj\loc,\obj\rem),x,\obj\rho$]
          \State $\obj{}\calS\gets\obj{}\calS \setminus \{x\}$
        \EndIf
      \EndMethod
    \end{algorithmic}
  \end{subfigure}
  \caption{Definition of the game between an adaptive AMQ and an
    adversary $\calA$.  The adversary gets $n$, $\epsilon$, and
    oracular access to $\calO$, which supports three operations:
    \calO.\amqlookup, \calO.\amqinsert, and \calO.\amqdelete.  The
    adversary wins if, after interacting with the oracle, it outputs
    an element $x'$ that is a false positive of the AMQ.  An AMQ is
    adaptive if there exists a constant $\epsilon<1$ such that no adversary
    wins with probability greater than $\epsilon$.}
  \label{fig:adaptivity-game}
\end{figure*}

\section{{\bf P{\small RELIMINARIES}}}
\seclabel{model}



We begin by defining the operations that our AMQ supports.  These
operations specify when it can access its local and remote states,
when it gets to update its states, how it receives feedback, and basic
correctness requirements (i.e., no false negatives).  We define
performance constraints (i.e., false-positive rates) later.

\newcommand{\todo}[1]{{\color{red} \scriptsize \sf #1}}


\begin{definition}[\defn{AMQs}]
  \deflabel{amqinternal}
  \deflabel{amq}%

  An approximate membership query data structure (AMQ) consists of the
  following deterministic functions. Here $\rho$ denotes the AMQ's
  private infinite random string, $\loc$ and $\rem$ denote its private
  local and remote state, respectively, $\calS$ represents the set of items
  that have been inserted into the AMQ more recently than they have
  been deleted, $n$~denotes the maximum allowed set size, and $\epsilon$
  denotes the false-positive probability.
  \begin{itemize}
  \item $\amqbuild[n,\epsilon,\rho]\longrightarrow (\loc,\rem)$.
    \amqbuild creates an initial state $(\loc,\rem)$.
  \item $\amqlookup[\loc,x,\rho] \longrightarrow (\loc',b)$. For
    $x\in\calU$, \amqlookup returns a new local state $\loc'$ and
    $b\in\{\present,\absent\}$.  If $x\in \calS$, then $b=\present$
    (i.e., AMQs do not have false negatives).  \amqlookup does not get
    access to $\rem$.
  \item \amqinsert[$(\loc,\rem),x,\rho$] $\longrightarrow
    (\loc',\rem')$.  For $|\calS|<n$ and $x \in \calU\setminus \calS$,
    \amqinsert returns a new state $(\loc',\rem')$.
    \amqinsert is not defined for $x\in\calS$.  \amqdelete is defined analogously.
  \item \amqupdate[$(\loc,\rem),x,\rho$] $\longrightarrow
    (\loc',\rem')$. For $x\not\in \calS$ such that
    $\amqlookup[\loc,x,\rho]=\present$, \amqupdate returns a new state
    $(\loc',\rem')$.
  \end{itemize}
\end{definition}

An AMQ is \defn{local} if it never reads or writes $\rem$;
an AMQ is \defn{oblivious} if
\amqupdate is the identity function on $(\loc,\rem)$.
Bloom filters, cuckoo filters, etc, are local oblivious AMQs.

\pparagraph{False positives and adaptivity.}  We say that $x$ is a
\defn{false positive} of AMQ state $(\loc,\rem)$ if $x\not\in \calS$ but
$\amqlookup[\loc,x,\rho]$ returns \present.

We define an AMQ's false-positive rate using the adversarial game in
\figref{adaptivity-game}.  In this game, we give the adversary
access to the AMQ via an oracle \calO.  The oracle keeps track of the
set $\calS$ being represented by the AMQ and ensures that the
adversary respects the limits of the AMQ (i.e., never overloads the
AMQ, inserts an item that is already in $\calS$, or deletes an item
that is not currently in $\calS$).  The adversary can submit queries
and updates to the oracle, which applies them to the AMQ and calls
\amqupdate whenever \amqlookup returns a false positive.  The
adversary cannot inspect the internal state of the oracle.
\adaptivityGame outputs \True iff the adversary wins, i.e.,  if,
after interacting with the oracle, $\calA$ outputs a false positive
$x'$ of the final state of the AMQ.


The \defn{static false-positive rate} is the probability, taken
over the randomness of the AMQ, that a particular $x\in\calU\setminus\calS$ is
a false positive of the AMQ.  This is equivalent to the probability
that an adversary that never gets to query the AMQ is able to output
a false positive.  We formalize this as follows.
An adversary is a \defn{single-query adversary} if it
never invokes \calO.\amqlookup.  We call this ``single-query'' because
there is still an invocation of $\calO.\amqlookup$ at the very end of the
game, when \adaptivityGame tests whether $x'$ is a false positive.

\begin{definition}
  An AMQ supports \defn{static false-positive rate $\epsilon$} if
  for all $n$ and all single-query adversaries $\calA$,
  \[
  \Pr[\adaptivityGame[\calA,n,\epsilon]=\True]\leq\epsilon.
  \]
\end{definition}

\begin{definition}
  An AMQ supports \defn{sustained false-positive rate $\epsilon$} if
  for all $n$ and all adversaries $\calA$,
  \[
  \Pr[\adaptivityGame[\calA,n,\epsilon]=\True]\leq\epsilon.
  \]
\end{definition}

An AMQ is \defn{adaptive} if there exists a constant $\epsilon<1$ such
that the AMQ guarantees a \fprate of at most $\epsilon$.

The following lemma shows that, since an adaptive AMQ accesses its remote
state rarely, it must use as much local space as a local AMQ.

\begin{lemma}\label{lem:adaptive-local-bits}
Any adaptive AMQ must have a local representation $\loc$ of size at
least $n\log (1/\epsilon)$. 
\end{lemma}
\begin{proof}
Consider an adaptive AMQ with a sustained false positive rate of
$\epsilon$.  Consider the local state $\loc'$ at the time when the
adversary provides $x'$.  By the definition of sustained-false positive
rate, $\loc'$ must have a static false positive rate of at most
$\epsilon$.  Thus, by the Bloom-filter lower
bound~\cite{CarterFlGi78,LovettPorat10}, $\loc'$ must have size at
least $n\log (1/\epsilon)$.
\end{proof}

\pparagraph{Cost model.}  We measure AMQ performance in terms
of the RAM operations on $\loc$ and in terms of the
number of 
updates and queries to the remote representation $\rem$.  We measure
these three quantities (RAM operations, remote updates, and remote
queries) separately.

We follow the standard practice of analyzing AMQ performance in terms
of the AMQ's maximum capacity, $n$.%
We
assume a word size $w = \Omega(\log u)$ in most of
the paper. For simplicity of presentation, we assume that $u=\poly(n)$
 but our results generalize.

\pparagraph{Hash functions.}\seclabel{hashing}
We assume that the adversary cannot find a
never-queried-before element that is a false positive of the AMQ with
probability greater than $\epsilon$. 
Ideal hash functions have this property for arbitrary adversaries.  If
the adversary is polynomially bounded, one-way functions are
sufficient to prevent them from generating new false
positives~\cite{NaorYo15}.  



\section{{\bf R{\small ESULTS}}}

We prove the following lower bound on the space required by an AMQ to maintain adaptivity.

\begin{restatable}{theorem}{lowerfinal} 
  \thmlabel{lowerfinal}
  Any adaptive AMQ storing a set of size $n$ from a universe of size
  $u > n^4$ requires $\Omega(\min\{n\log n,n\log\log u\})$ bits of
  space whp to maintain any constant \fprate $\epsilon < 1$.
\end{restatable}

Together, \defref{amqinternal}, \thmref{lowerfinal} and
\lemref{adaptive-local-bits} suggest 
what an optimal adaptive AMQ should look like.
\lemref{adaptive-local-bits} says that $\loc$ must have at least
$n\log (1/\epsilon)$ bits.  
\thmref{lowerfinal}
implies that any adaptive AMQ with $\loc$ near this lower bound must
make remote accesses.

A consequence of \defref{amqinternal} is
that AMQs access $\rem$ only when the system is accessing $\dict$, so,
if an AMQ performs $O(1)$ updates of $\rem$ for each update of
$\dict$ and $O(1)$ queries to $\rem$ for each query to $\dict$, then
accesses to $\rem$ are asymptotically free.   Thus, our target is an
AMQ that has approximately $n\log(1/\epsilon)$ bits in $\loc$ and
performs $O(1)$ accesses to $\rem$ per update and query.

Our upper bound result is such an adaptive AMQ:

\begin{theorem}
  \thmlabel{broom-performance}
  There exists an adaptive AMQ---the broom filter---that,  for any \fprate
  $\epsilon$ and maximum capacity $n$,   
  attains the following performance: 
  \begin{itemize}
    \item \textbf{Constant local work:} \label{item:broom-local-work} $O(1)$ 
    operations for inserts,  deletes, and lookups w.h.p.
  \item \textbf{Near optimal local space:} \label{item:broom-space}
    $(1 + o(1)) n \log\frac{1}{\epsilon} + O(n)$ local space w.h.p.\footnote{All logarithms in this paper
are base $2$ unless specified otherwise.}
  \item \textbf{Asymptotically optimal remote
      accesses:} \label{item:broom-remote-access}  $O(1)$ updates to
    $\rem$ for each delete to $\dict$; $O(1)$ updates to $\rem$ with
    probability at most $\epsilon$ for each insertion to $\dict$;
    $O(1)$ updates to $\rem$ for each false positive.

  \end{itemize}
\end{theorem}  

The local component of the broom filter is, itself, an AMQ with
performance that strictly dominates the Bloom Filter, which requires
$(\log e) n \log (1/\epsilon)$ space and $O(\log (1/\epsilon))$ update
time~\cite{Bloom70}, and matches (up to lower-order
terms) or improves upon the performance of more efficient
AMQs~\cite{PaghPaRa05,FanAnKa14,BenderFaJo12,Porat09}.

Since $\loc$ contains an AMQ, one way to interpret our results is that
a small local AMQ cannot be adaptive if it is only informed of true
positives versus false positives, but it can adapt if it is given a
little more information.  In the case of the broom filter, it is given
the element of $S$ causing a false positive, that is, the element in
$S$ that has a hash function collision with the query, as we see next.

\section{{\bf B{\small ROOM} F{\small ILTERS}: D{\small EFINING} F{\small INGERPRINTS}}}
\seclabel{broom-fingerprints}

The broom filter is a single-hash-function AMQ
\cite{PaghPaRa05,BenderFaJo12,FanAnKa14}, which means that it 
stores fingerprints for each element in $\calS$. In this section, 
we begin our proof of \thmref{broom-performance} by describing what
fingerprints we store and how they  
establish the sustained false-positive rate of broom filters.  In
\secref{broom-cases}, we show how to maintain the fingerprints space-efficiently and 
in $O(1)$ time.

\subsection{\bf\emph{Fingerprints}}

The broom filter has a hash function
$h: \calU\rightarrow \{ 0,\ldots n^c \}$ for some constant
$c\geq 4$. 
Storing an entire hash takes $c\log n$ bits, which is too much
space---we can only afford approximately
$\log(1/\epsilon)$ bits per element.
Instead, for set $\calS = \{y_1,y_2,\ldots,y_n\}$, 
the broom filter stores a set of \defn{fingerprints}
$\calP = \{p(y_1), p(y_2),\ldots, p(y_n) \}$, where each $p(y_i)$
is a \defn{prefix} of $h(y_i)$, denoted
$p(y_i) \prefix h(y_i)$. 

\pparagraph{Queries.}  
A query for $x$ returns \present iff there
exists a $y\in\calS$ such that $p(y) \prefix h(x)$. 
The first
$\log n+ \log (1/\epsilon)$ bits of a fingerprint comprise the
\defn{\baseline fingerprint}, which is subdivided as in a quotient
filter~\cite{BenderFaJo12,PandeyBeJo17}. In particular, the first $q=\log n$ bits
comprise the \defn{quotient}, and the next $r=\log (1/\epsilon)$ bits
the \defn{remainder}.  The remaining bits (if any) comprise the
\defn{adaptivity bits}.

\pparagraph{Using the parts of the fingerprint.}  The \baseline
fingerprint is long enough to guarantee that the false-positive
rate is at most $\epsilon$.  We add adaptivity bits to fix
false positives, in order to achieve a sustained
false-positive rate of $\epsilon$. Adaptivity bits are also added during
insertions. We maintain the following invariant: 

\begin{invariant}\label{inv:hash-invariant}
No fingerprint is a prefix of another.
\end{invariant}

By this invariant, a query for $x$ can match at most one $p(y)\in\calP$. 
As we will see, we can fix a false positive by adding
adaptivity bits to the single $p(y)$, for which $p(y) \prefix h(x)$.
Thus, adding adaptivity bits during insertions reduces the number of
adaptivity bits added during false positives, which will allow us to
achieve $O(1)$ work and remote accesses for each operation.

Shortly we will give a somewhat subtler reason why adaptivity bits are added during insertions---in order to defeat deletion-based timing attacks on the sustained false-positive rate.

\pparagraph{Maintaining the fingerprints.}
Here we describe what the broom filter does on a call to \extend.  
In this section we drop $(\loc,\rem)$ and $\rho$ from the notation for simplicity.

We define a subroutine of \extend which we call \mextend$(x,\calP)$.  This function is used to maintain \invref{hash-invariant} and to fix false positives.

Observe that on a query $x$ there exists at most one $y$ for which $p(y) \prefix h(x)$, by \invref{hash-invariant}.  
If such a $y$ exists, the $\mextend(x,\calP)$ operation modifies the local representation by appending adaptivity bits to $p(y)$ until $p(y) \not\prefix h(x)$.  (Otherwise, $\mextend(x,\calP)$ does nothing.)
Thus, \mextend{} performs remote accesses to \rev, where \arev{x} returns the (unique)
$y\in\calS$ such that $p(y)\prefix h(x)$.  \rev is a part of $\rem$, and can be
implemented using a dictionary.

We can define $\extend(x)$ as follows: 
  \begin{itemize}
\item {\bf Queries.} 
If a query $x$ is a false positive,  we call \mextend$(x,\calP)$, after
which $x$ is no longer a false positive.  
\item {\bf Insertions.}
When inserting an element $x$ into $\calS$, we first check if
\invref{hash-invariant} is violated, that is, if there exists a $y\in
\calS$ such that $p(y) \prefix h(x)$.\footnote{This step and the
  following assume $x$ does not already belong
  to $\calS$.  If it does, we don't need to do anything during insertions.}
If so, we call 
\mextend$(x,\calP)$, after
which $p(y) \not\prefix h(x)$. 
 Then we add the shortest prefix of
$h(x)$ needed to maintain  \invref{hash-invariant}.
\item {\bf Deletions.} 
  Deletions do not make calls to \extend. 
  We defer the details of the deletion operation until after we
  discuss how to reclaim bits introduced by \extend.  For now we note
  the na\"\i{}ve approach of deleting an element's fingerprint is
  insufficient to guarantee a sustained false-positive rate.
 \end{itemize}

\subsection{\bf\emph{Reclaiming Bits}} \seclabel{reclaiming-bits} 
Each call to \extend adds bits, and so we
need a mechanism to remove bits.
An amortized way to reclaim bits is to rebuild the broom filter
with a new hash function every $\Theta(n)$ calls to \extend.

This change from old to new hash function can be deamortized without losing a
factor of~$2$ on the space.  We keep two hash functions, $h_a$ and $h_b$; any
element $y$ greater than frontier $z$ is hashed according to $h_a$, otherwise,
it is hashed according to $h_b$.  At the beginning of a \defn{phase}, frontier
$z=-\infty$ and all elements are hashed according to $h_a$.  Each time we call
\extend, we delete the smallest constant $c>1$ elements in $\calS$ greater than
$z$ and reinsert them according to $h_b$.  (Finding these elements requires
access to $\rem$; again this can be efficiently implemented using standard data
structures.) We then set $z$ to be the value of the largest reinserted element.
When $z$ reaches the maximum element in $\calS$, we begin a new phase by
setting $h_a = h_b$, picking a new $h_b$, and resetting $z = -\infty$.  We use
this frontier method for deamortization so that we know which hash function to
use for queries: lookups on $x\leq z$ use $h_b$ and those on $x>z$ use $h_a$.

\begin{observation}\obslabel{rehash}
A hash function times out after  $O(n)$ calls to \extend.
\end{observation}

Because every call to \extend introduces an expected constant number of adaptivity
bits, we obtain:

\begin{lemma}\lemlabel{linear-extend-bits}
  In any phase, \extend introduces $O(n)$ adaptivity bits into the
  broom filter with high probability.
\end{lemma}

\ifjournal
\begin{proof}
  By~\obsref{rehash}, for some constant $c_1$, there are $c_1 n$ false
  positives before the entire AMQ gets rehashed. Constant $c_1$ is
  determined by the number of elements that get rehashed per false
  positive, and so can be tuned.

  Each time there is a false positive, there is a collision with
  exactly one element by \invref{hash-invariant}. Given 
  that there is a collision, the probability that it can be
  resolved by extending fingerprints by $i$ bits is $2^{-i}$. Whenever
  an element is rehashed, its adaptivity bits get thrown out.  Thus,
  by Chernoff bounds, the number of adaptivity bits in the data
  structure at any time is $O(c_1 n)$ w.h.p.
\end{proof}
\fi

If we did not have deletions, then \obsref{rehash} and
\lemref{linear-extend-bits} would be enough to prove a bound on
total size of all fingerprints---because adaptivity bits are removed
as their hash function times out.  To support
deletions we introduce adaptivity bits via a second mechanism.
We will show that this second mechanism also introduces a total of
$O(n)$ adaptivity bits per phase.

\subsection{\bf\emph{Deletions and Adaptivity Bits}}\seclabel{deletions-adaptivity-bits}

It is tempting to support deletions simply by removing fingerprints
from $\calP$, but this does not work.  To see why, observe that false
positives are eliminated by adding adaptivity bits.  Removing
fingerprints destroys history and reintroduces false positives. This
opens up the data structure to timing attacks by the adversary.

We describe one such timing attack to motivate our solution.  The
adversary finds a false positive $x$ and an element $y \in \calS$
that collides with $x$.  (It finds $y$ by deleting and reinserting
random elements until $x$ is once again a false positive.)  The 
attack then consists of repeatedly
looking up $x$, deleting $y$, then inserting $y$.  This results in a false positive on every lookup until $x$ or $y$'s hash function changes.

Thus, the broom filter needs to remember the history for deleted
elements, since they might be reinserted. Only once $y$'s hash
function has changed can $y$'s history be forgotten. A profligate
approach is to keep the fingerprints of deleted elements as ``ghosts''
until the hash function changes. Then, if the element is reinserted,
the adaptivity bits are already there.  Unfortunately, remembering
deleted elements can blow up the space by a constant factor, which we
cannot afford.

Instead, we remember the adaptivity bits and quotient from each
deleted element's fingerprint---but we forget the remainder.  Only once
the hash function has changed do we forget everything.
This can be accomplished by including deleted elements in the strategy described in \secref{reclaiming-bits}.
(with deletions, we increase the requirement on  adaptivity bits reclaimed at once to $c > 2$). 

Now when a new element $x$ gets inserted, we check whether there
exists a ghost that matches $h(x)$.  If so, then we give $x$ at least
the adaptivity bits of the ghost, even if this is more than needed to
satisfy \invref{hash-invariant}. This scheme guarantees the following: 

\begin{property}\proplabel{inherit-adaptivity-bits}
  If $x$ is a false positive because it collides with $y$, then it
  cannot collide with $y$ again until $x$ or $y$'s hash function times
  out (even if $y$ is deleted and reinserted).
\end{property}

\subsection{\bf\emph{Sustained False-Positive Rate}}
We now establish the sustained false-positive rate of broom filters.
We begin by introducing notation:

\begin{definition}
  Hashes $h(x)$ and $h(y)$ have a \defn{soft collision} when they have
  the same quotient.
  They have a \defn{hard collision} when 
  they have the same quotient and remainder. 
  Hash $h(x)$ and fingerprint $p(y)$ have a \defn{full collision} if
  $p(y) \prefix h(x)$.
\end{definition}

The hash function is fixed in this section, so we refer to $x$ and $y$ themselves as having (say) a soft collision, with the understanding that it is their hashes that collide.

\begin{lemma}\lemlabel{collision-probability}
  The probability that any query has a hard collision with any of $n$
  fingerprints is at most $\epsilon$.
\end{lemma}

\begin{proof}
  The probability that any query collides with a single fingerprint is
  $2^{-(\log n + \log{(1/\epsilon)})} = \epsilon/n$. Applying the
  union bound, we obtain the lemma.
\end{proof}

\begin{lemma}\lemlabel{sustained-fp-rate}
  The sustained false-positive rate of a broom filter is $\epsilon$.
\end{lemma}

\begin{proof}
  We prove that on any query $x \notin S$, 
  $\Pr[\exists y\in\calS \mathrel{}$ $|$ $\mathrel{} x$ has a full collision with
  $y]\leq \epsilon$, regardless of the previous history.  Any previous
  query that is a negative or a true positive has no effect on the
  data structure.  Furthermore, deletions do not increase the chance
  of any full collision, so we need only consider false positives and
  insertions, both of which induce rehashing.

  We say that $x\in\calU$ and $y\in\calS$ are \defn{related at time $t$}
  if (1) there exists $t'<t$ such that $x$ was queried at time $t'$ and
  $y$ was in $\calS$ at $t'$, and (2) between $t'$ and
  $t$, the hash functions for $x$ and $y$ did not change.
  Suppose $x$ is queried at time $t$.  Then, by
  \propref{inherit-adaptivity-bits}, if $x$ and $y$ are related at
  time $t$, then $\Pr[ x \text{ is a false positive at } t] =
  0$.  If $x$ and $y$ are not related at time $t$, then 
  $\Pr[ x \text{ has a full 
    collision with } y ] \leq 
\Pr[ h(x) \text{ has a hard
    collision with } h(y) ]$.  Finally, by 
 \lemref{collision-probability}, 
  $\Pr[ x \text{ is a false positive at } t] \leq \epsilon$. 
\end{proof}
\subsection{\bf\emph{Space Bounds for Adaptivity Bits}}

We first prove that at any time there are~$O(n)$ adaptivity bits.
Then we bootstrap this claim to show a stronger property: there are
$\Theta(\log n)$ fingerprints associated with
$\Theta(\log n)$ contiguous quotients, and these fingerprints have a
total of~$O(\log n)$ adaptivity bits w.h.p.\ (thus they
can be stored in~$O(1)$ machine words).

For the purposes of our proofs, we partition adaptivity bits into two
classes: \defn{extend bits}, which are added by calls to \mextend, and
\defn{copy bits}, which are added on insertion due to partial matches
with formerly deleted items.  As some bits
may be both extend and copy bits, we partition adaptivity bits by
defining all the adaptivity bits in a fingerprint to be of the same
type as the last bit, breaking ties in favor of extend.
If an item is deleted and then reinserted, its bits are of the same
type as when it first got them.  (So if an item that gets extend bits
is deleted and reinserted with the same adaptivity bits, then it still
has extend bits.)

\begin{lemma}\lemlabel{linear-adaptivity-bits}
  At any time, there are $O(n)$ adaptivity bits in the broom filter
  with high probability.
\end{lemma}
\ifjournal
\begin{proof}
  \lemref{linear-extend-bits} bounds the number of extend bits. We
  still need to bound the number of copy bits.  
  We do so using
   a straightforward application of Chernoff bounds.

  The number of quotients that have at least
  $k$ extend bits is $O(n/k)$. This is because the total number of
  extend bits is $O(n)$.  Therefore,
  the probability that $h(x)$ accumulates $k$ extend bits is
  $O(1/(k2^k)$.  (This is the probability that $h(x)$ matches a
  quotient with $k$ extend bits times the probability that those
  extend bits match.)

  Thus, the expected number of copy bits from length-$k$ strings
  is $O(n/2^k)$, for $1\leq k\leq \Theta(\log n)$. 
  By Chernoff,
    these bounds also hold w.h.p. for $k \leq (\log n)/\log\log n$; for $k > (\log n)/\log\log n$ Chernoff bounds give that there are $O(\log n)$ bits from length-$k$ strings w.h.p.  
    Thus, the total number of adaptivity
  bits is $O(n)$, w.h.p.
\end{proof}
\fi

\begin{lemma}\lemlabel{logn-group-extend-bits} 
  There are~$\Theta(\log n)$ fingerprints associated with a range of~$\Theta(\log n)$ contiguous quotients, and these fingerprints have~$O(\log n)$ 
total extend bits w.h.p.
\end{lemma} 

%
%

\ifjournal
\begin{proof}
  As long as there are~$O(n)$ adaptivity bits and~$\Theta(n)$ stored elements,
  then no matter how the adaptivity bits are distributed:
  the first time that $x$ is queried or inserted with hash function
  $h$, \extend is called with probability
  $\Theta(\epsilon)$. By Chernoff bounds, before the phase~(see
  \secref{reclaiming-bits}) ends, there are $O(n/\epsilon)$ distinct
  elements not in $\calS$ that are ever queried and $O(n/\epsilon)$
  distinct elements that are ever inserted into $\calS$.

  We can now calculate an upper bound on the number of adaptivity bits at any
  time $t$. Recall that at the very beginning of phase $\ell$, there is a
  unique hash function $h_\ell$ that is in use, because $h_{\ell-1}$ has
  expired, and $h_{\ell+1}$ has not been used yet. Any extend adaptivity bits
  that are in the broom filter at time $t$ in phase $\ell$ were generated as a result
  of collisions generated by $h_{\ell-1}$, $h_{\ell}$, or $h_{\ell+1}$.

  Now consider all elements that were ever inserted or queried any
  time during phase $\ell-1$, $\ell$, or $\ell+1$ with $h_{\ell-1}$,
  $h_\ell$, or $h_{\ell+1}$.  If we took all these elements, and
  inserted them one at a time into $\calS$, calling \extend to resolve
  any collisions, this scheme would at least generate all the extend adaptivity
  bits that are present at time $t$.

  It thus suffices to show that even with this overestimate, the
  fingerprints associated with a range of $\Theta(\log n)$
  contiguous quotients have a total of $O(\log n)$ extend bits w.h.p.
  Call the $\Theta(\log n)$ quotients under consideration the \defn{group}.
  Define $0/1$-random variable $X_i= 1$ iff element $x_i$ lands in
  the group and induces a call to extend.  Thus,
  $\Pr[X_i= 1]\leq O(\epsilon\log n/n)$.  There are $O(n/\epsilon)$
  elements inserted, deleted, and/or queried in these rounds.  Thus,
  by Chernoff bounds, the number of elements that land in this
  quotient group is $O(\log n /\epsilon)$, and at most $O(\log n)$ of
  them get adaptivity bits w.h.p.

  We bound the number of bits needed to resolve the
  collisions.  There are~$O(\log n)$ elements that land in this group.  We 
  model this as a balls and bins game, where elements land in the same
  bin if they share the same quotient and remainder.  Let random
  variable $K_i$ represent the number of elements in the~$i$th
  nonempty bin.  The expected number of bits that get added until
  all collisions are resolved is~$2\sum_{i=1}^{O(\log n)} \log(K_i)$. By the
  convexity of the $\log$ function,  
  $\sum_{i=1}^{O(\log n)} \log(K_i) = O(\log n) $, regardless of the
  distribution of the elements into bins.

  To achieve concentration bounds on this result, we upper bound this process
  by a different process.  Each time we add a bit, there is a probability of at
  least $1/2$ that it matches with at most half of the remaining strings. Thus,
  the number of adaptivity bits is stochastically dominated by the number of
  coin flips we need until we get $\Theta(\log n)$ heads, which is $\Theta(\log
  n)$ w.h.p. 
\end{proof}
\fi

\begin{lemma}\lemlabel{logn-group-copy-bits} 
  There are~$\Theta(\log n)$ fingerprints associated with a range of~$\Theta(\log n)$ contiguous quotients, and these fingerprints have
$O(\log n)$ total adaptivity bits w.h.p.
\end{lemma}
\ifjournal
\begin{proof}
  We established the bound on extend bits in
  \lemref{logn-group-extend-bits}; now we focus on copy bits.
  
  Consider any time $t$ when there are $n$ elements in the broom
  filter, and consider any group of $\Theta(\log n)$ contiguous
  quotients.  By Chernoff bounds, $\Theta(\log n)$ of these $n$
  elements have hashes that have a soft collision with  one of these quotients w.h.p.  By
  \lemref{linear-adaptivity-bits}, there are a total of $O(\log n)$
  extend bits in this range.  We now show that there are also a total
  of $O(\log n)$ copy bits.

  The scheme from \secref{deletions-adaptivity-bits} can be described
  in terms of balls and bins as follows.  There are $\Theta(\log n)$
    bins, one for each quotient. Each \defn{string} of adaptivity bits belongs
  in a bin. Some bins can have multiple strings (but by standard
  balls-and-bins arguments, the fullest bin has
  $O(\log n/\log\log n)$ strings of adaptivity bits).
  When a new element $x$ is inserted, it lands in the bin determined
  by $h(x)$. Then $p(x)$ inherits the adaptivity bits in the bin iff
  $h(x)$ matches those adaptivity bits. (This means that any given
  string of adaptivity bits started out as extend bits, even if it got
    copied many times as copy bits.) 

  We now bound the number of adaptivity bits by considering a variation that
  adds more bits than the scheme from \secref{deletions-adaptivity-bits}.  For
  each element inserted into a bin, we keep appending copy bits as long as
  there is a match with some string of adaptivity bits in the bin.  Once there
  is a mismatch with every string, we stop.  Thus, while the scheme from
  \secref{deletions-adaptivity-bits} adds copy bits only on \emph{complete}
  matches, we allow \emph{prefix} matches while still retaining good bounds.

  We again overestimate the bounds by assuming that the
  adaptivity bits are adversarially (rather than randomly) divided into bit
  strings and that the bit strings are adversarially distributed among the
  bins.

  Let random variable $K_i$ denote the number of adaptivity bit strings in the
  bin where the $i$th element lands. 
  The first claim that we want to make is
  the following:

  \emph{Claim. } $\Pr[K_i\geq X] < O(1/X)$.

  \emph{Proof. } This follows from Markov's inequality and
  \lemref{logn-group-extend-bits}.  Since w.h.p., the total number of
  adaptivity bits is at most $O(\log n)$, the expected number of bits
  in a bin, and thus the expected number of strings, is $O(1)$.

  We next show the following claim, one of the cornerstones of the
  proof.
    
  \emph{Claim. } $\sum_{i=1}^{\Theta(\log n)}\log(K_i)=O(\log n)$.
    
  \emph{Proof. }
  By the previous claim, 
  \begin{equation*}
    \begin{split}
      \Pr[K_i\geq X]\leq \Pr[\text{we flip a coin and get at least} \\
      \text{$\log (X) - O(1) $ tails before any head}].
    \end{split}
  \end{equation*}
  Therefore, the
  probability that $\sum_{i=1}^{c\log n}\log(K_i)= d\log n$ is at most
  the probability that we flip a coin $d\log n$ times and get at most
  $c\log n$ heads.  For a suitable choice of constants $c$ and $d$,
  this is polynomially small.

  Next we bound the total number of adaptivity bits that the elements
  inherit. Element $x_i$ lands in a bin with $K_i$ adaptivity bit
  strings. Each time a bit is added, with probability at least $1/2$,
  the number of adaptivity strings that still match with $h(x_i)$
  decreases by half. Specifically, suppose that $k$ adaptivity strings
  still match $x_i$.  With probability at least $1/2$, after the next
  bit reveal, at most $\lfloor k/2\rfloor$ still match.  So after an
  expected $\leq 2\log(K_i)$ bits, no adaptivity bit strings still
  match $x_i$.  Once again this game is modeled as flipping a coin
  until until we get $\Theta(\log n)$ heads, and by Chernoff, only
  $\Theta(\log n)$ are needed w.h.p.\
\end{proof}
\fi

\section{{\bf B{\small ROOM} F{\small ILTERS}: I{\small MPLEMENTING} F{\small INGERPRINTS}}}\label{sec:broom-cases}

In \secref{broom-fingerprints}, we showed how to use
  fingerprints to achieve a sustained false-positive rate of
  $\epsilon$. 
In this section we give space- and time-efficient implementations for
the fingerprint operations that are specified in
\secref{broom-fingerprints}.  We explain how we store and
manipulate adaptivity bits (\secref{encoding-adaptivity-bits}),
quotients (\secref{encoding-quotients}), and remainders.  We
describe two variants of our data structure, because there are two
ways to manage remainders, depending on whether
$\log(1/\epsilon)\leq 2\log\log n$, the \defn{small-remainder case}
(\secref{small-remainder}), or
$\log(1/\epsilon)> 2\log\log n$, the \defn{large-remainder case}
(\secref{large-remainder}).

\pparagraph{Bit Manipulation within Machine
  Words.} 
In~\secref{wordstuff}, we show how to implement a
variety of primitives on machine words in $O(1)$ time using word-level
parallelism.
The upshot is that from now
on, we may assume that the asymptotic complexity for any operation on
the broom filter is simply the number of machine words that are
touched during the operation.

\ifjournal
In \secref{wordstuff}, we show how to implement a
variety of primitives on machine words in $O(1)$ time using word-level
parallelism; see \lemref{parallelsearch}.  
The upshot is that from now
on, we may assume that the asymptotic complexity for any operation on
the broom filter is simply the number of machine words that are
touched during the operation.
\fi

\subsection{\bf\emph{Encoding Adaptivity Bits and Deletion Bits}}
\seclabel{encoding-adaptivity-bits}\subseclabel{encoding-adaptivity-bits}

We store adaptivity bits separately from the rest of the fingerprint.
By \lemref{logn-group-copy-bits}, all of the adaptivity bits in any
range of $\Theta(\log n)$ quotients fit in a constant number of words.
Thus, all of the searches and updates to 
\ifjournal
(both copy and extend)
\fi
adaptivity bits take $O(1)$ time.

\subsection{\bf\emph{Encoding Quotients}}
\seclabel{encoding-quotients}\subseclabel{encoding-quotients} 

Quotients and remainders are stored succinctly in a scheme similar to quotient
filters~\cite{BenderFaJo12,PandeyBeJo17}; we call this high-level scheme \defn{quotienting}.

Quotienting stores the \baseline  fingerprints succinctly in an
array of $\Theta(n)$ slots, each consisting of $r$ bits.  
Given a fingerprint with quotient $a$
and remainder $b$, we would like to store $b$ in position $a$ of the array. 
This allows us to reconstruct the fingerprint based on $b$'s location. 
So long as the number of slots is not much more than the number of stored quotients,
this is an efficient representation. (In particular, we will have a sublinear number of extra slots in our data structure.) 

The challenge is  that  multiple fingerprints may have the
same quotient and thus contend for the same
location. Linear probing is a standard technique for resolving
collisions: slide an element forward in the
array until it finds an empty slot.  Linear probing does not
immediately work, however, since the quotient is supposed to be
reconstructed based on the location of a remainder.  The
quotient filter implements linear probing by maintaining a
small number (between 2 and 3) of metadata bits per array slot 
which encode the target slot for a remainder even when it
 is shifted to a different slot.

 The standard quotient filter does not achieve constant time
 operations, independent of $\epsilon$.  This is because when the
 remainder length $r=\log (1/\epsilon)=\omega(1)$, and the fingerprint
 is stored in a set of $\Omega(\log n)$ contiguous slots, there can be
 $\omega(1)$ locations (words) where the target fingerprint could be.
 (This limitation holds even when the quotient filter is half empty,
 in which case it is not even space efficient enough for
 \thmref{broom-performance}.)

 Nonetheless, the quotient filter is a good starting point for the
 broom filter because it allows us to maintain a multiset of \baseline
 fingerprints subject to insertions, deletions, and queries.  In
 particular, some queries will have a hard collision with multiple
 elements.%
 \footnote{This is the main challenge in achieving
   optimality with the single-hash function bloom filters of Pagh et
   al.~\cite{PaghPaRa05} or the backyard hashing construction of
   Arbitman et al.~\cite{ArbitmanNaSe10}. Instead we used techniques
   that permit the same element to be explicitly duplicated multiple
   times.  } We need to compare the adaptivity bits of the query to
 the adaptivity bits of each colliding element.  The quotienting
 approach guarantees that these adaptivity bits are contiguous,
 allowing us to perform multiple comparisons simultaneously using
 word-level parallelism.  In particular,
 \lemref{logn-group-extend-bits} ensures that the adaptivity bits for
 $O(\log n)$ quotients fit into $O(1)$ machine words.

\subsection{\bf\emph{Broom Filter Design for the Small-Remainder Case}}
\seclabel{small-remainder}

In this section we present a data structure for the case that $r = O(\log\log n)$.

\pparagraph{High Level Setup.} 
Our data structure consists of a primary and a secondary level.   
Each level is essentially a quotient filter; however, we slightly change the insert and delete operations for the primary level in order to ensure constant-time accesses.  

As in a quotient filter, the primary level consists of $n(1 + \alpha)$ slots, where each slot has a remainder of size $r = \log (1/\epsilon) = O(\log\log n)$.  
Parameter~$\alpha$ denotes the subconstant extra space we leave in our data structure; thus the primary level is a quotient filter as described in \secref{encoding-quotients}, with space parameterized by $\alpha$ (and with slightly modified inserts, queries, and deletes). We require $\alpha \geq \sqrt{(9r\log\log n)/\log n}$.  

The secondary level consists of a quotient filter with $\Theta(n/\log n)$
slots with a different hash function $h_2$.  Thus, an
element $x$ has two fingerprints $p_1(x)$ and $p_2(x)$.  The internals of the
two levels are maintained entirely independently:
Invariant~\ref{inv:hash-invariant} is maintained separately for each level, and
adaptivity bits do not carry over from the primary level to the secondary
level.

\pparagraph{How to Perform Inserts, Queries and Deletes.}
To insert $y\in\calS$, we first try to store the fingerprint $p_1(y)$ in the primary level.  This uses the technique described in \secref{encoding-quotients}: we want to store the remainder in the slot determined by the quotient.  If the slot is empty, we store the remainder of $p_1(y)$ in that slot.  Otherwise, we begin using linear probing to look for an empty slot, updating the metadata bits accordingly; see~\cite{BenderFaJo12,PandeyBeJo17}.

However, unlike in previous quotienting-based data structures, we stop our probing for an empty slot early: 
the data structure only continues the linear probing over $O((\log n)/r)$ slots (and thus $O(1)$ words).  If all of these slots are full, the item gets stored in the secondary level.
In~\lemref{secondlevel} we show that it finds an empty slot in~$O(1)$ words in the secondary level w.h.p.

We always attempt to insert into the primary level first.  In particular, even if $x$ is deleted from the secondary level while reclaiming bits~(\secref{reclaiming-bits}), we still attempt to insert $x$ into the primary level first.

Queries are similar to inserts---to query for $y$, we calculate $p_1(y)$ and search for it in the primary level for at most $O((\log n)/r)$ slots; if this fails we calculate $p_2(y)$ and search for it in the secondary level. 

\begin{lemma}
\lemlabel{primarytosecondary}
With high probability, $O(n/\log^2 n)$ elements are inserted into the secondary level.
\end{lemma}

\begin{proof}
    Partition the primary level into \defn{primary bins}
    of $(1 + \alpha)(\log n)/r$ consecutive slots.  
    An element is inserted into the secondary level only if it is inserted into a sequence of $\Omega((\log n)/r)$ full slots; for this to happen either the primary bin containing the element is full or the bin adjacent to it is full.  
 We bound the number of full primary bins.
    
    In expectation, each bin is (initially) hashed to by $(\log n)/r$ elements.  Thus, by Chernoff bounds, the probability that a given primary bin is hashed to by at least $(1 + \alpha)(\log n)/r$ elements is at most 
    $\exp(-(\alpha^2 \log n)/(3r)) \leq 1/\log^3 n$.

    Thus, in expectation, $n/\log^3 n$ primary bins are full.  Since these events are negatively correlated, we can use Chernoff bounds, and state that $O(n/\log^3 n)$ primary bins are full with high probability.  
    
    Each primary bin is hashed to by $O(\log n)$ elements in expectation (even fewer, in fact).  
    Using Chernoff, each primary bin is hashed to by $O(\log n)$ elements w.h.p. 

    Putting the above together, even if all $O(\log n)$ elements hashed into any of the $O(n/\log^3 n)$ overflowing primary bins (or either adjacent bin) are inserted into the secondary level, we obtain the lemma.
\end{proof}

\begin{lemma}
    \lemlabel{secondlevel}
With high probability, all items in the secondary level are stored at most $O(\log n/r)$ slots away from their intended slot.  
\end{lemma}

\begin{proof}
Partition the secondary level into \defn{secondary bins} of $\Theta(\log n/r)$ consecutive slots.
Thus, there are $\Theta(nr/\log^2 n)$ secondary bins.  The lemma can only be violated if one of these bins is full.

By \lemref{primarytosecondary}, we are inserting $O(n/\log^2 n)$ elements into these bins.  
By classical balls and bins analysis, because there are more bins than balls, 
    the secondary bin with the most balls has $O((\log n)/\log\log n) = O((\log n)/r)$ elements with high probability.  Thus, no secondary bin ever fills up with high probability.
\end{proof}

\pparagraph{Performance.}  The $O(1)$ lookup time follows by definition in the primary level, and by \lemref{secondlevel} in the secondary level.  The total space of the primary level is $O((1 + \alpha)n\log (1/\epsilon)) + O(n)$, and the total space of the second level is $O((n\log (1/\epsilon))/\log n)$.  We guarantee adaptivity using the \extend function defined in \secref{broom-fingerprints}, which makes $O(1)$ remote memory accesses per insert and false positive query.


\subsection{\bf\emph{Broom Filter for  Large Remainders}}
\seclabel{large-remainder}\subseclabel{large-remainder}

In this section we present a data structure for the
the large-remainder case, $\log (1/\epsilon) > 2\log\log n$. 
Large remainders are harder to store efficiently since
only a small number can  fit in a machine word.  E.g., we
are no longer guaranteed to be able to store the remainders from all
hard collisions in $O(1)$ words w.h.p.

However, large remainders also have advantages.  We are very likely to be able to search using only a small portion of the remainder---a portion small enough that many  can be packed into $O(1)$ words.
In particular, we can ``peel off'' the first $2\log\log n$ bits of the
remainder, filter out collisions just based on those bits, and we are
left with few remaining potential collisions.  We call these
\defn{partial collisions}.  

So we have an initial check for uniqueness, then a remaining check for the rest
of the fingerprint. This allows us to adapt the small-remainder case to handle
larger remainders without a slowdown in time.

\pparagraph{Data structure description.} As before, our data structure consists of two parts. We refer to them as the \defn{primary level} and the \defn{backyard}.  This notation emphasizes the structural difference between the two levels and the relationship with backyard hashing~\cite{ArbitmanNaSe10}. Unlike the small-remainder case, we use only a single hash function.

The primary level consists of two sets of slots: \defn{signature slots} of size $2\log\log n$, and \defn{remainder slots} of size $r - 2\log\log n$.  
As in~\secref{small-remainder}, the number of remainder slots is $(1 + \alpha)n$ and the number of signature slots is $(1+\alpha)n$, where $\alpha \geq \sqrt{18\log^2\log n/\log n}$.  Because the appropriate slot is found while traversing the signature slots, we only need to store metadata bits for the signature slots; they can be omitted for the remainder slots.  The signature slots are stored contiguously; thus $O(\log n/\log\log n)$ slots can be probed in $O(1)$ time.

Each item is stored in the same remainder slot as in the normal quotient filter (see \subsecref{encoding-quotients}).  The signature slots mirror the remainder slots; however, only the first $2\log\log n$ bits of the remainder are stored, the rest are stored in the corresponding remainder slot.

\pparagraph{The primary level.} To insert an element $y$, we first try to
insert $p(y)$ in the primary level. 
We find the signature slot corresponding to the quotient of $p(y)$.  We then search through at most $O(\log n/\log\log n)$ signatures to find a partial collision (a matching signature) or an empty slot. We use metadata bits as usual---the metadata bits guarantee that we only search through signatures that have a soft collision with $p(y)$. 

If there is a partial collision---a signature that matches the first $2\log\log n$ bits of the remainder of $p(y)$---we insert $p(y)$ into the backyard.  If there is no empty slot, we insert $p(y)$ into the backyard.  
If we find an empty slot but do not find a partial collision, we insert $p(y)$ into the empty slot; this means that we insert the signature into the empty signature slot, and insert the full remainder of $p(y)$ into the corresponding remainder slot. We update the metadata bits of the signature slots as in~\cite{BenderFaJo12,PandeyBeJo17}.

Querying for an element $x$ proceeds similarly.  In the primary level, we find the signature slot corresponding to the quotient of $p(x)$.  We search through $O(\log n/\log\log n)$ slots for a matching signature.  If we find a matching signature, we look in the corresponding remainder slot to see if we have a hard collision; 
if so we return \present.  If we do not find a matching signature, or if the corresponding remainder slot does not have a hard collision, we search for $p(x)$ in the back yard. 

\pparagraph{The back yard.}
The back yard is a compact hash table that can store $O(n/\log n)$ elements 
with $O(1)$ worst-case insert and delete time~\cite{DemainedePa06,ArbitmanNaSe10}.
When we store an element $y$ in the back yard, we store its
entire hash $h(y)$. Thus, w.h.p.\ there are no collisions in
the back yard.  Since the back yard has a capacity for
$\Theta(n/\log n)$ elements, and each hash has size $\Theta(\log n)$, the
back yard takes up~$\Theta(n)$ bits, which is a lower-order term. 

\begin{lemma}
    \lemlabel{smalleps}
    The number of elements stored in the back yard is $O(n/\log^2 n)$ with high probability. 
\end{lemma}
\ifjournal
\begin{proof}
  An element is stored in the backyard only if 1. it is in a sequence of $\Omega(\log n/\log\log n)$ full slots, or 2. it has a partial collision with some stored element.

  The number of elements that are in a sequence of full slots is $O(n/\log^2 n)$ with high probability; this follows immediately from \lemref{primarytosecondary} with $r = 2\log\log n$.

  A query element $x$ has a partial collision with an element $y$ if they have the first $\log n + 2\log\log n$ bits of their fingerprint in common.  Thus, $x$ and $y$ collide with probability $1/(n\log^2 n)$; thus $x$ has a partial collision with $1/\log^2 n$ stored elements in expecation.  The lemma follows immediately from Chernoff bounds.
\end{proof}
\fi

\ifjournal
\pparagraph{Performance.}

The back yard requires $O(n)$ total space, since each hash is of length $O(\log n)$.  The primary level requires $(1 + \alpha) nr$ space for all primary slots, plus $O(n)$ extra space for the adaptivity bits stored as in \subsecref{encoding-adaptivity-bits}.

Inserts, deletes, and queries require $O(1)$ time.  The search for partial collisions involves $O(\log n/\log\log n)$ signature slots, which fit in $O(1)$ words; these can be searched in constant time.
We look at a single remainder slot, which takes $O(1)$ time.  If needed, any back yard operation requires $O(1)$ time as well. 

\fi


\section{{\bf A L{\small OWER} B{\small OUND} {\small ON} A{\small DAPTIVE} AMQ{\small S}}}
\seclabel{lower}
 
In this section, we show that an AMQ cannot maintain adaptivity along
with space efficiency.  More formally, we show that any adaptive AMQ
must use $\Omega(\min\{n\log n,n\log\log u\})$ bits.  This means that
if an AMQ is adaptive and the size of $\loc$ is
$o(\min\{n\log n,n\log\log u\})$ bits, then it must access $\rem$.
The proof itself does not distinguish between bits stored in $\loc$ or
$\rem$.  For convenience, we show that the lower bound
holds when all bits are stored in $\loc$; this is equivalent to
lower bounding the bits stored in $\loc$ and $\rem$.  

Interestingly, a similar lower bound was studied in the context of Bloomier
filters~\cite{ChazelleKiRu04}.  
\ifjournal
The Bloomier filter is an AMQ designed to solve the problem of storing
$n$ items for which it must return \present, along with a whitelist of
$\Theta(n)$ items for which it must return \absent.  Other queries must have a
static false-positive rate of $\epsilon$.  Chazelle et
al.~\cite{ChazelleKiRu04} give a lower bound on any data structure that updates
this whitelist dynamically, showing that such a data structure must use
$\Omega(n\log\log (u/n))$ space.  Their lower bound implies that if the adversary gives an
AMQ a dynamic white list of false positives that it needs to \emph{permanently} fix, then it must use too much space.  
\fi
In this section, we generalize
this bound to all adaptive AMQ strategies.

\subsection{\bf\emph{Notation and Adversary Model}}

We begin by further formalizing our notation and 
defining the adversary used in the lower bound.
We fix $n$ and $\epsilon$ and drop them from most notation.  We use
$\build(\calS,\rho)$ to denote the state that results from calling
$\amqbuild(n,\epsilon,\rho)$ followed by $\amqinsert(x,\rho)$ for each
$x\in \calS$ (in lexicographic order).

\pparagraph{Adversary Model.} 
The adversary does not have
access to the AMQ's internal randomness $\rho$, or any internal state  $\lmem$ of the
AMQ.  The adversary can only issue a query $x$ to the AMQ and only learns the
AMQ's output---\present or \absent---to query $x$.

The goal of the adversary is to adaptively generate a sequence of $O(n)$
queries and force the AMQ to either use too much space or to fail to satisfy a
\fprate of
$\epsilon$.

Let $\epsilon_0 = \max \{1/n^{1/4}, (\log^2\log u)/\log u\}$.  Our lower
bound is $m = |\lmem| =  \Omega(n\log 1/\epsilon_0)$.
Note that $\epsilon_0 \leq \epsilon$; otherwise the classic AMQ lower
bound of $m \geq 
n\log 1/\epsilon$~\cite{CarterFlGi78,LovettPorat10} is sufficient to
prove~\thmref{lowerfinal}.  One can think of $\epsilon_0$ as a maximum
bound on the effective false positive rate---how often the AMQ
encounters elements that need fixing.

\pparagraph{Attack Description.} 
First, the adversary chooses a set $\calS$ of size $n$ uniformly at random from $\calU$.  
%
Then, the attack proceeds in rounds.
The adversary selects a set $Q$ of size $n$ uniformly at random from
$\calU-\calS$. Starting from $Q$, in each round, he queries the elements that
were false positives in the previous round. To simplify analysis, we assume that
the adversary orders his queries in lexicographic order.  Let $\FP_i$ be the
set of queries that are false positives in round $i\ge 1$.  The attack:
\begin{enumerate}
\item In the first round, the adversary queries each element of $Q$.  
\item In round $i> 1$, if $|\FP_{i-1}|> 0$, the adversary queries 
each element in $\FP_{i-1}$; otherwise the attack ends.  
\end{enumerate}

\pparagraph{Classifying False Positives.}
The crux of our proof is that some false positives are difficult to
fix---in particular, these are the queries where an AMQ is unable to
distinguish whether or not $x \in \calS$ by looking at its state
$\lmem$.\footnote{This is as opposed to easy-to-fix queries where, e.g., the
AMQ answers \present randomly to confuse an adversary.  For all
previous AMQs we are aware of, all false positives are absolute false
positives.} 
We call $y\in\calU\setminus\calS$ an \defn{\absfp} of a state $\lmem$ and
randomness $\rho$ if there exists a set $\calS'$ of
size $n$ and a sequence of queries $(x_1, \ldots, x_t)$ such that $y \in
\calS'$ and $\lmem$ is the state of the AMQ 
when queries $x_1,\ldots,x_t$ are performed 
on $\build(\calS',\rho)$. 
We use
$\absset(\lmem,\calS,\rho)$ to denote the set of \absfps of state $\lmem$, randomness $\rho$, and true-positive set $\calS$.
We call $(\calS', (x_1, \ldots, x_t))$ a \defn{witness} to $y$.  

We
call $y \in\calU \setminus \calS$ an \defn{\oabsfp} of $\calS$ and $\rho$ 
if and only if $y\in \absset(\build(\calS,\rho),\calS, \rho)$.
We denote the set of \oabsfps{} $\oabsset(\calS, \rho) =
\absset(\build(\calS,\rho),\calS, \rho)$.

As the AMQ handles queries, it will need to fix some previous false positives.
To fix a false positive, the AMQ must change its state so that it can
safely answer \absent to it.  For a state $\lmem$, we define the set of
elements that are no longer false positives by the set $\text{FIX}(\lmem,\calS,\rho)
= \oabsset(\calS,\rho)\setminus \absset(\lmem,\calS,\rho)$.  Note that all fixed false
positives are \oabsfps.

As an AMQ cannot have false negatives, it cannot fix an \oabsfp{}
$y$ unless it learns that $y\notin \calS$.  This is formalized in
the next two observations.
\begin{observation}\obslabel{nofalseneg} 
For any randomness $\rho$, set $\calS$, and state $\lmem$ of the AMQ, 
if a query $x \in \absset(\lmem,\calS,\rho)$, then
\amqlookup$(\lmem,x,\rho)$ must return \present.
\end{observation}

\begin{observation} \obslabel{needquerytofix} 
Let $\lmem_1$ be a state of the AMQ before a query $x$
and $\lmem_2$ be the updated state after $x$ (that is, after invoking $\amqlookup$ and possibly $\amqupdate$).  Let $y$ be an \absfp{} of $\lmem_1$ with
witness $S_y$.  Then if $y$ is not an \absfp{} of $\lmem_2$, then $x\in S_y$.
\end{observation}

\subsection{\bf\emph{Analysis}}

We start with an overview of the lower bound.

First, we recall a known result~(Claim~\ref{claim:original-afp}) that a space-efficient
AMQ must start with a large number of \oabsfps for almost all $\calS$.
Given that an AMQ has a large number of \oabsfps, an adversary can
discover a fraction of them through randomly chosen queries
$Q$~(\lemref{counting-afp-queries}).

Next, we show that through adaptive queries, the adversary forces the AMQ to fix almost all of these
discovered \oabsfps, for most sets $Q$~(%
    \lemref{winnowing} and
\lemref{double-whp}).

The crux of the proof relies on~\lemref{probfixabs}, which says that
the AMQ cannot fix too many \emph{extra} \oabsfps during the
attack---thus, it needs a large number of distinct ``fixed'' sets to
cover all the different sets of original absolute false positives that
the adversary forces the AMQ to fix.  This is where we use that the
AMQ only receives a limited amount of feedback on each
false positive---it cannot fix more false positives without risking
some false negatives.

Finally, we bound lower bound the space used by the AMQ by observing that there is a 1-to-1 mapping from
``fixed'' sets of original absolute false positives to AMQ states.
Thus, we can lower bound the number of AMQ states (and hence
the space needed to represent them) by lower-bounding
the number of sets of original absolute false positives the adversary
can force the AMQ to fix.

\begin{observation} \obslabel{distinctrepresentations} For a given randomness
$\rho$ and set $\calS$ of size $n$, consider two fixed false positive sets $\textnormal{FIX}(\lmem_1,\calS,\rho)$
and $\textnormal{FIX}(\lmem_2,\calS,\rho)$.  Then if
$\textnormal{FIX}(\lmem_1,\calS,\rho)\neq \textnormal{FIX}(\lmem_2,\calS,\rho)$, then $\lmem_1
\neq \lmem_2$.
\end{observation}


\pparagraph{Discovering \oabsfps through random queries.}
%
While for some special sets $\calS$ given in advance, an AMQ may be able to store $\calS$ very accurately 
(with very few false positives), this is not true for most random
sets $\calS$ chosen from the universe by the adversary. 
We note the following claim from Naor and Yogev~\cite{NaorYo15}.

\begin{claim}[{\cite[Claim 5.3]{NaorYo15}}]\label{claim:original-afp}
Given any randomness $\rho$ of AMQ using space $m \le n \log 1/\epsilon_0 + 4n$ bits, for any set $\calS$ of size $n$ chosen
uniformly at random 
from $\calU$, we have: 
$\Pr_{\calS} \left[ |\oabsset(\calS,\rho)| \le u\epsilon_0 \right] \le 2^{-n}$.
\end{claim}

For the remainder of this section, we fix a set $\calfS \subseteq \calU$ of size $n$ such that 
$|\oabsset(\calfS,\rho)| > u \epsilon_0$.\footnote{%
  With probability $1/2^n$, the adversary gets unlucky
and chooses a set $\calfS$ that does not satisfy this property, in which case he fails. This is okay, because
we only need to show \emph{existence} of a troublesome set $\calfS$---and we in fact show the stronger claim that most $\calfS$ suffice.}
Let $\calQ$ be the set of all possible query sets $Q$ the adversary can choose, that is, $\calQ = \{ Q \subseteq \calU\setminus\calS^* \mbox{ $|$ } |Q| = n\}$.  (We do not include $S^*$ in the notation of $\calQ$ for simplicity.)
The following lemma follows immediately from Chernoff bounds.
\begin{lemma}\lemlabel{counting-afp-queries}
  For a fixed randomness $\rho$ of an AMQ of size $m\leq n\log 1/\epsilon_0 + 4n$ and fixed set $\calfS$ such that
  $|\oabsset(\calfS,\rho)| > u \epsilon_0$, we have 
  $\Pr_{Q\in\calQ} \left[|Q \cap \oabsset(\calfS,\rho)| = \Omega(n\epsilon_0) \right] \ge 1-1/\poly(n).$
\end{lemma}


\pparagraph{Forcing the adaptive AMQ to fix large number of \oabsfps.}
From the definition of \fprate, the AMQ must fix an $\epsilon$ fraction of false positives in expectation in each round.  If the expected number of false positives that the
AMQ has to fix in each round is high, classic concentration bounds imply that the AMQ must fix close to this expected number with high probability in each round.  This implies that there must be a round where the AMQ fixes a large number of \oabsfps.  The next lemma formalizes this intuition.

For a given $Q$, let $\forced(Q, \calfS,\rho)$ be the maximal-sized set of query elements (out of $Q$)
that the AMQ has fixed simultaneously in any state.  
For $1\leq i\leq t$, let $\lmem_i$ be the state of the AMQ after query $x_i$. 
Then we let 
$\forced(Q,\calfS,\rho) = \fix(\lmem_{t'},\calfS,\rho)$ for the smallest $t'$ such that $|\forced(Q,\calfS,\rho)| \geq \fix(\lmem_{t''},\calfS,\rho)$ for any $t''$.

The following lemma shows that the AMQ must, at the beginning of some round in the first $O(n)$ queries by the adversary, fix $\widetilde{\Omega}(n\epsilon_0)$ false positives.  

\begin{lemma}
    \lemlabel{winnowing}
Consider an AMQ of size $m \leq n \log 1/\epsilon_0 + 4n$.
    For any set $Q$ satisfying $Q\cap\oabsset(\calfS,\rho) = \Omega(n\epsilon_0)$, 
    there exists a round $T(Q,\rho)$ and a state $\lmem_{T(Q,\rho)}$ 
  at the beginning of round $T(Q,\rho)$ such that 
  $|\fix(\lmem_{T(Q,\rho)}, \calfS,\rho)| = \Omega(n\epsilon_0/\log_{\epsilon}{\epsilon_0})$ w.h.p., that is,  
\begin{align*}
  &\Pr_\rho \biggl[|\forced(Q,\calfS,\rho)| = \Omega(n\epsilon_0 / \log_{\epsilon} \epsilon_0)
~\bigg|~ \\ 
&\quad\qquad |Q\cap\oabsset(\calfS,\rho)| = \Omega(n\epsilon_0) \biggr] 
       \geq 1 - 1/\poly(n).
\end{align*}
  Round $T(Q,\rho)$ is reached in at most $O(n)$ total queries.  

\end{lemma}
\ifjournal
\begin{proof}
  We fix $Q$ and set $T = T(Q,\rho)$.

    Recall that $\FP_T$ denotes the set of queries that are false positives in round $T$, and let $T_{f} = \log_{\epsilon} \epsilon_0$. 
Since the AMQ has a \fprate of $\epsilon$, we have $|\FP_1| = O(n\epsilon)$. As $\epsilon \ge \epsilon_0 \ge 1/n^{1/4}$, by Chernoff bounds,
      we have $|\FP_{T+1}| \le \epsilon |\FP_T|(1 + 1/\log n)$ with high
probability for all $1 \le T\le T_{f}$.

Suppose there does not exist a round $T < T_{f}$ such that the lemma holds,
that is, in each round $T < T_{f}$, $|\fix(\lmem_T, \calfS,\rho)| \leq
n\epsilon_0/2\log_{\epsilon}{\epsilon_0}$, where $\lmem_T$ is the state of the
AMQ at the beginning of round $T$.  In round $T_{f}$, the AMQ is asked
$|\FP_{T_{f}-1}|\leq (\epsilon(1 + 1/\log n))^{T_{f}-2}n = O(n
\epsilon_0)$ queries. From our assumption,  $|\oabsset(\calfS,\rho) \cap
\FP_{T_{f}-1}| \geq n\epsilon_0(1 - 1/2\log_{\epsilon} \epsilon_0)^{T_{f}
- 1} = \Omega(n \epsilon_0)$.

To maintain a \fprate of $\epsilon$, it must hold that $|\FP_{T_{f}}| = O(n \epsilon
\epsilon_0)$ with high probability. Thus, in round $T_f$ the AMQ must answer \absent to
$\Omega(n (1-\epsilon) \epsilon_0) = \Omega (n \epsilon_0)$ \oabsfps{} from the
set $|(\oabsset(\calfS,\rho) \cap \FP_{T_{f}-1}) \setminus \FP_{T_{f}}|$.
We denote this set of \oabsfps{} queries that the AMQ says \absent to in
round $T_{f}$ as $\calA_{T_{f}}$.

Let $\lmem_{T_{f}, x}$ denote the state of the AMQ in round
$T_{f}$ just before query $x$ is made. Then by~\obsref{nofalseneg}, $x \in
\fix(\lmem_{T_{f}, x}, \calfS,\rho)$ for any  $x \in \calA_{T_{f}}$.  We now show that
all $x \in \calA_{T_{f}}$ must simultaneously be in the set of fixed false
positives of the state $\lmem_{T_{f}}$ at the beginning of round
$T_{f}$.  Note that $ x \in \oabsset(\calfS,\rho) \cap \FP_{T_{f-1}}$ and all
queries between query $x$ in round $T_{f-1}$ and query $x$ in round
$T_{f}$ are distinct from  $x$ and were chosen independently from $x$ in
round 1. As there can be at most $n$ queries in between query $x$ in
consecutive rounds, using~\obsref{needquerytofix}, the probability that there
exists a state $\lmem_i$ between $\lmem_{T_{f-1},q}$ and
$\lmem_{T_{f},q}$ such that $x \notin \fix(\lmem_i, \calfS,\rho)$ is at most $n^2/u <
1/n^2$.  Thus, with high probability, $x \in \fix(\lmem_{T_{f}},
\calfS,\rho)$ for any given $x \in \calA_{T_{f}}$. That is, $\calA_{T_{f}}
\subseteq \fix(\lmem_{T_{f}}, \calfS,\rho)$, and thus,  $ |\fix(\lmem_{T_{f}}, \calfS,\rho)|=
\Omega(n\epsilon_0)$.

Furthermore, round $T$ is reached in $n + \sum_{i=1}^{T} \FP_i \leq n(1 -
\epsilon^T)/(1 - \epsilon) = O(n)$ queries.  
\end{proof}
\fi

    For simplicity, let $\epsilon_0' = \epsilon_0/\log_\epsilon \epsilon_0$.  (This does not affect our final bounds.)
The next lemma follows from \lemreftwo{counting-afp-queries}{winnowing} and shows that (for most $\rho$), most query sets $Q$ satisfy \lemref{winnowing}
with high probability.
\begin{lemma}\lemlabel{double-whp} Given an AMQ of size~$m \leq n\log 1/\epsilon_0 + 4n$ and set $\calfS$
such that  $|\oabsset(\calfS,\rho)| \geq u\epsilon_0$, 
for all but a $1/\poly(n)$ fraction of $Q$, there
  exists a round~$T(Q,\rho)$ such that the AMQ is
forced to fix~$\Omega(n \epsilon_0')$ \oabsfp queries w.h.p. over $\rho$, that is,
  \begin{equation*}
    \begin{split}
    \Pr_{\rho} \left[\Pr_{Q\in\calQ} \left[|\forced(Q,\calS^*,\rho)| = \Omega(n\epsilon_0')\right] \geq 1 - \frac{1}{\poly(n)} \right]\\
      \geq 1 - 1/\poly(n). 
    \end{split}
  \end{equation*}
\end{lemma}

\ifjournal
\begin{proof} 
\begin{align*}
    & \Pr_{\rho} \left[  \Pr_{Q\in\calQ} \left[|\forced(Q, \calfS,\rho)| = \Omega \left( n\epsilon_0' \right) \right] \geq 1 - \frac{1}{\poly(n)}\right] \notag \\
& \qquad \ge \Pr_{\rho} 
    \left[  
    \begin{aligned}
        \Pr_{Q\in\calQ} \left[|\forced(Q, \calfS,\rho)| =  \Omega \left( n\epsilon_0' \right)  ~\bigg|~ |Q \cap \absset(\calfS,\rho)| = \Omega(n \epsilon_0)\right]\cdot \qquad\qquad\\
    \cdot \Pr_{Q\in\calQ} \left[ |Q \cap \absset(\calfS,\rho)|= \Omega(n \epsilon_0)  \right] 
        \geq 1 - \frac{1}{\poly(n)}
    \end{aligned}
    \right] \\
& \qquad \ge \Pr_{\rho} \left[  \Pr_{Q\in\calQ} \left[|\forced(Q, \calfS,\rho)| =  \Omega \left( n\epsilon_0' \right)  ~\bigg|~ |Q \cap \absset(\calfS,\rho)| = \Omega(n \epsilon_0)\right] \geq 1 - \frac{1}{\poly(n)} \right] \label{eq:1} \\
    & \qquad \ge \left(1 - \frac{1}{\poly(n)}\right)
\end{align*}
    The second step is from \lemref{counting-afp-queries}, and the final step is from \lemref{winnowing}.
\end{proof}
\fi

\pparagraph{Small AMQs cannot fix too many \oabsfps.} 
Next, we show that 
for a randomly-chosen $y$, knowing $y\notin\calfS$ is unlikely to give information to the AMQ about which set $\calfS$ it is storing.
In particular, an AMQ may try to rule out false
positives that may be correlated to a query $y$. 
\ifjournal
For example, an AMQ may (without asymptotic loss of space) partition the universe into pairs of elements such that if it learns one item in a pair is a false positive, it is guaranteed that the other is as well.  
With such a strategy, the AMQ can succinctly fix two false positives per query instead of one.  
Could a more intricate strategy allow the AMQ to fix more false positives---even enough to be an obstacle to our lower bound?
\fi
We rule out this possibility in the following lemma. On a random query, the AMQ is very unlikely to fix a given false positive.  We use this to make a w.h.p. statement about the number of fixed elements using Markov's inequality in our final proof.

\begin{lemma} \lemlabel{probfixabs}
Let $(x_1,\ldots,x_t)$ be a sequence of queries taken from a uniformly-sampled random set $Q\subset \calU$
of size $n$, and let $\lmem'$ be the state of the AMQ after these
queries. For an element $y \in \calU$,
the probability over $Q$ that $y$ is a fixed false positive after these queries is
$n^2/u$. That is, $\Pr_{Q\in\calQ} \left[ y \in \textnormal{FIX}(\lmem',\calfS,\rho)
\right]  \leq n^2/u.$
\end{lemma} 

\begin{proof} If $y \in \textnormal{FIX}(\lmem',\calfS)$, then by \obsref{needquerytofix}, for any witness set $S'$ of $y$,
$S'\cap Q$ must be nonempty.  Fix a single witness set $S'$.  For a given
$s'\in S'$ and $x'\in Q$, $\Pr(s' = x') = 1/u$.  Taking the union bound over
all $n^2$ such pairs $(s',x')$ we achieve the lemma.  
\end{proof}

\pparagraph{Final lower bound.}
We prove the desired lower bound.

\lowerfinal*

\begin{proof}    Assume by contradiction that $m \leq n\log 1/\epsilon_0 + 4n$.
\ifjournal
Recall that $\epsilon_0 = \max \{1/n^{1/4}, (\log^2\log u)/\log u\}$).  
\fi
Applying
\obsref{distinctrepresentations}
  to the state $\lmem_{T(Q,\rho)}$ from \lemref{winnowing}
  we obtain 
  \[ 
  2^m \geq  \left| \{\textnormal{FIX}(\lmem_{T(Q,\rho)},\calfS,\rho) ~|~ Q \subset \calU, |Q|=n\} \right|.  \] 
  We lower bound how many distinct fixed sets the
AMQ needs to store for a given $\rho$.  
 
   By definition, $\forced(Q,\calfS,\rho)$ $\subseteq \textnormal{FIX}$$(\lmem_{T(Q,\rho)}, \calfS,\rho)$.  But, the set of fixed elements cannot 
be much bigger than
the set of fixed queries.  Consider an arbitrary $x\in \calU$.  By
\lemref{probfixabs}, $x$ is a fixed false positive with probability at most
  $n^2/u$. Thus, $\E_{Q\in\calQ}$ $[|\textnormal{FIX}(\lmem_{T(Q,\rho)},\calfS,\rho), \calfS,\rho)|]$ = $n^2$.  By
Markov's inequality,
  $\Pr_{Q\in\calQ}$ $\big[|\textnormal{FIX}(\lmem_{T(Q,\rho)}$, $\calfS,\rho)|= O(n^2)\big]$  = $\Omega(1)$.  
  Then there exists a set\footnote{$\calQ^*$ is a function of $\rho$ and $\calS$; we omit this to simplify notation.} 
  $\calQ^*$ $= \big\{Q \subseteq \calU {~\big|~} |Q| = n$, $\textnormal{FIX}$ $(\lmem_{T(Q,\rho)},\calfS,\rho)$ = $O(n^2)\big\}$
  such that  $|\calQ^*| = \Omega(|\calQ|)$.
  Thus, 
    $ \left|  \left\{ \textnormal{FIX}\left(\lmem_{T(Q,\rho)}, \calfS,\rho \right) ~{|}~  Q \in\calQ \right\} \right| \geq \scalebox{.9}{$\left| \left\{\forced(Q,\calfS,\rho) {|} Q \in \calQ^*, |\forced(Q,\calfS,\rho)| = \Omega(n\epsilon_0') \right\}\right|$}/{ \binom{O(n^2)}{\Omega(n\epsilon_0')} }.$
    
    Now, we count the number of distinct $\forced(Q,\calfS,\rho)$.  
Let $\calZ = \{Z \subseteq \oabsset(\calfS,\rho) \mbox{ $|$ } |Z| = \Theta(n \epsilon_0')\}$.
  We show that with high probability a randomly chosen set $Z \in \calZ$ belongs to the set of $\forced(Q, \calfS,\rho)$ for some $Q$ (because we choose $Q$ uniformly at random, this probabilistic argument lower bounds the number of such $Z$ immediately). Recall that $|\oabsset(\calfS,\rho)|> u\epsilon_0$. 
\begin{align}
  &\Pr_{Z\in\calZ} [ \exists Q\in\calQ^* \text{ with } Z \subseteq \forced(Q, \calfS,\rho)]\notag \\
&\qquad\geq\Pr_{\substack{ Z\in\calZ\\Q\in\calQ*\text{ with } Q\supset Z}} [Z\subseteq \forced(Q,\calfS,\rho) \}]\\
    & \qquad \geq \Pr_{\substack{ Z\in\calZ\\Q\in\calQ*\\ Q\supset Z}} \left[\scalebox{.9}{$Z\subseteq \forced(Q,\calfS,\rho)$} \bigg| \scalebox{.9}{$|\forced(Q,\calfS,\rho)| = \Omega(n\epsilon_0')$} \right] \notag\\
&\qquad\qquad \cdot \Pr_{Q\in\calQ*} \bigg[ |\forced(Q,\calfS,\rho)| = \Omega(n\epsilon_0')~\bigg|~ \notag\\
  &\qquad\qquad\qquad\qquad |Q\cap\oabsset(\calfS,\rho)| = \Omega(n\epsilon_0)\bigg]  \cdot \Pr_{Q\in\calQ*} \left[|Q\cap\oabsset(\calfS,\rho)| = \Omega(n\epsilon_0)\right] \label{eq2:secondterm}\\ 
    &\qquad \geq \left( 1\middle/\binom{n}{\Omega(n\epsilon_0')}\right)\left(1 - \frac{1}{\poly(n)}\right).\notag
\end{align}
For the first term in step~(\ref{eq2:secondterm}), note that
if the AMQ is forced to fix $\Omega(n \epsilon_0')$ queries out of query set $Q$, then the probability that
a randomly chosen $Z$ corresponds to this forced query set is at most 1/(total number of possible forced subsets of $Q$).
  There can be at most $\binom{n}{\Omega(n\epsilon_0')}$ such subsets.
  The second and third term in step~(\ref{eq2:secondterm}) follow from \lemref{double-whp} and~\lemref{counting-afp-queries} respectively---because $|\calQ^*| = \Omega(|\calQ|)$, a simple probabilistic argument shows that the probability over $Q\in\calQ^*$ rather than $Q\in\calQ$ retains the high probability bounds. 
We lower bound the total number of distinct forced sets (ignoring the lower-order $1-1/\poly(n)$ term) as
  \[
    \left| \left\{\forced( Q,\calfS,\rho) ~{\big |}~ Q\in\calQ^*\right\}\right| \geq  \left. {\binom{u\epsilon_0 }{\Omega(n \epsilon_0')}} \middle/{\binom{n}{\Omega(n\epsilon_0')}} \right. .
   \]

    Putting it all together%
    \ifjournal
    and removing the less-than-constant $1/\poly(n)$ term, 
    \fi
\[
  2^m \geq \left. {\binom{u\epsilon_0 }{\Omega(n \epsilon_0')}}\middle/{\binom{n}{\Omega(n\epsilon_0')}} \binom{O(n^2)}{\Omega(n\epsilon_0')} \right. .
\]
Taking logs and simplifying, 
\ifjournal
 (recall $(x/y)^y \leq \binom{x}{y}\leq (xe/y)^y$),
\fi
\[
    m  =\Omega\left(n\epsilon_0'\left(\log{ \frac{u \log(1/\epsilon_0)}{n}}  - \log \frac{n}{\epsilon_0'} \right)\right).
\]
    We have $\log (u \log(1/\epsilon_0) /n) - \log(n/\epsilon_0') = \Omega(\log u)$ 
    \ifjournal
    because $u \gg n^2/\epsilon_0'$.  Because $\epsilon$ is a constant, 
    \fi
    $\epsilon_0' = \Omega(\epsilon_0/\log (1/\epsilon_0))$.  
    Thus, $m  = \Omega\big(\frac{n\epsilon_0 \log u}{\log(1/\epsilon_0)}\big)$.   
 
Using the definition of $\epsilon_0$, we have two cases. 
\begin{enumerate}
\item\label{firstcase} If $1/n^{1/4} \geq (\log^2 \log u)/\log u$, then $m = \Omega(n \log n \log u /\log\log n)$.
\item\label{secondcase} If $1/n^{1/4} < (\log^2 \log u)/\log u$, then $m = \Omega (n \log \log u)$.
\end{enumerate}
    In case \ref{firstcase}, we get a contradiction to the assumption that $m \leq n\log 1/\epsilon_0 + 4n$.  In case \ref{secondcase} if $m \leq n\log\log u + 4n$, then we obtain a bound of $\Omega(n\log\log u)$. 
\end{proof}

\pparagraph{Matching upper bound.}~%
\ifjournal
We give an AMQ construction that shows that the above bound is tight. 
\fi

\begin{lemma}
    There exists an adaptive AMQ that can handle $O(n)$ adaptive queries using $O(\min\{n\log n,n\log\log u\})$ bits of space w.h.p.
\end{lemma}

\ifjournal
\begin{proof}
    If $\log n = O(\log\log u)$, we build a standard bloom filter with \fpprob $\epsilon = 1/n^c$ for $c > 1$.  The adversary never queries an \oabsfp w.h.p.
(If he does, the AMQ can store it without affecting the w.h.p. space bound.)

    Otherwise, we have $\log\log u = o(\log n)$, and thus $n/\log u = n^{\Omega(1)}$.  We build a standard bloom filter with $\epsilon =\log\log u/\log u$; this requires $O(n\log\log u)$ space.  
    
    In $O(n)$ queries there can be $O(n\log\log u/\log u)$ false positives \ifjournal; from the assumption on the size of $u$ this is $n^{\Omega(1)}$ so Chernoff bounds imply $O(n\log\log u/\log u)$ false positives 
\fi  w.h.p.
We store all false positives in a whitelist. 
\ifjournal 
Each requires $O(\log u)$ bits.  
\fi
Thus, the space used is $O(n\log\log u)$.
\end{proof}
\fi

\section{\bf A{\small DDITIONAL} R{\small ELATED} W{\small ORK}}\seclabel{related}

Bloom filters~\cite{Bloom70} are a
ubiquitous data structure for approximate membership queries and have inspired
the study of many other AMQs; for e.g., see~\cite{ChazelleKiRu04, PutzeSaSi07,
PaghPaRa05,ArbitmanNaSe10, BenderFaJo12}.  Describing all Bloom filter variants
is beyond the scope of this paper; see surveys~\cite{BroderMi04,TarkomaRoLa12}.
Here, we only mention a few well-known AMQs in~\secref{amqiary}
and several results closely related to adaptivity in~\secref{related-adaptivity}.

\subsection{\bf\emph{AMQ-iary}}\seclabel{amqiary} 

\pparagraph{Bloom Filters~\cite{Bloom70}.}
The standard Bloom filter, representing a set $\calS \subseteq \calU$,
is composed of $m$ bits, and uses $k$ independent hash functions
$h_1, h_2, \ldots, h_k$, where $h_i: \calU  \rightarrow \{1, \ldots, m\}$.
To insert $x \in \cal{S}$, the bits $h_i (x)$ are set to $1$ for $1 \le 1 \le k$.
A query for $x$ checks if the bits $h_i (x)$ are set to $1$ for $1 \le i\le k$---if they are,
it returns ``{present}'', and else it returns ``{absent}''. A Bloom filter does not support deletes. 
For a \fpprob of
$\epsilon > 0$, it uses $m=(\log e) n \log (1/\epsilon)$ bits  
and has an expected lookup time of $O(\log (1/\epsilon))$.

\pparagraph{Single-Hash-Function AMQs~\cite{PaghPaRa05,CarterFlGi78}.} Carter et
al.~\cite{CarterFlGi78} introduced the idea of an AMQ that uses a single hash
function that maps each element in the universe to one of $\lceil n/\epsilon \rceil$
elements, and then uses a compressed exact-membership tester.  
Pagh et al.~\cite{PaghPaRa05} show how to build the first near-optimal single hash-function AMQ by applying a universal hash
function $h: U \rightarrow \{0,1,\ldots, \lceil n/\epsilon \rceil\}$ to
$\calS$, storing the resulting values using $(1 + o(1)) n \log{ \frac
1\epsilon} + O(n)$ bits. Their construction achieves $O(1)$ amortized insert/delete bounds and has been deamortized, in the case of inserts,
by the backyard hashing construction of Arbitman et al.~\cite{ArbitmanNaSe10}.

\pparagraph{Quotient Filters~\cite{BenderFaJo12,
    PandeyBeJo17}.\seclabel{prelim-qf}}
The quotient filter (QF) is a practical variant
of Pagh et al.'s single-hash function AMQ~\cite{PaghPaRa05}.  
The QF serves as the basis for our adaptive AMQ and is described in \secref{broom-fingerprints}.

\pparagraph{Cuckoo Filter~\cite{FanAnKa14}}  The cuckoo filter is also a practical variant of Pagh et al.'s single-hash function AMQ~\cite{PaghPa08}.  However, the implementation is based on cuckoo hashing rather than linear probing.  
This difference leads to performance advantages for some parameter settings~\cite{FanAnKa14,PandeyBeJo17}.  

It would be interesting to develop a provably adaptive variant of the cuckoo filter.  Its structure makes it difficult to use our approach of maintaining adaptivity bits . On the other hand, 
analyzing the Markov chain resulting from the heuristic approach to cuckoo-filter adaptivity in~\cite{MitzenmacherPo17} has its own challenges.

\subsection{\bf\emph{Previous Work on Adaptive AMQs and Adaptive Adversaries}}\seclabel{related-adaptivity} 
\pparagraph{Bloomier filters.} Chazelle et al.~\cite{ChazelleKiRu04}'s {\em
Bloomier filters} generalize Bloom filters to avoid a predetermined list of
undesirable false positives.  Given a set $S$ of size $n$ and a whitelist $W$
of size $w$, a Bloomier filter stores a function $f$ that returns {``present''}
if the query is in the $S$, {``absent''} if the query is not in $S \cup W$, and
``is a false positive'' if the query is in $W$. Bloomier filters use $O( (n+w)
\log{1/\epsilon})$ bits.  The set $S \cup W$ cannot be updated without
blowing up the space used and thus their data structure is limited to a static
whitelist.

\pparagraph{Adversarial-resilient Bloom filters.}
Naor and Yogev~\cite{NaorYo15} study Bloom filters in the context of a
{repeat-free adaptive adversary}, which 
queries  elements until it can find a never-before-queried element that has a false-positive probability greater than $\epsilon$.
(They do not consider the false-positive probability of repeated queries.)
They show how to protect an AMQ from a repeat-free adaptive
adversaries using cryptographically secure hash functions so that new
queries are indistinguishable from uniformly selected
queries~\cite{NaorYo15}. Hiding data-structure hash functions has been studied beyond AMQs (e.g., see~\cite{GerbetKuLa2015,HardtWo13,MironovNaSe11, BlumFuKe93, NaorYo13}).

\pparagraph{Adaptive cuckoo filter.} 
Recently, Mitzenmacher et al.~\cite{MitzenmacherPo17} introduced the {\em adaptive cuckoo filter} which removes false positives after they have been queried.
They do so by consulting the remote representation of the set stored in a hash table. Their data structure takes more space than a regular cuckoo filter
but their experiments show that it performs better on real packet traces and when queries have a temporal correlation.

\pparagraph{Other AMQs and false-positive optimization.}
{\em Retouched Bloom filters}~\cite{DonnetBa10} and 
{\em generalized Bloom gilter}~\cite{LauferVe05} reduce the number of false positives by introducing false negatives.
Variants of Bloom filter that choose the number of hash functions assigned to an element based on its frequency in some predetermined query distribution
have also been studied~\cite{BruckGaJi06, ZhongLu08}.

\pparagraph{Data Structures and Adaptive Adversaries} 
Nonadaptive data structures lead to significant problems in
security-critical applications---researchers have described
{algorithmic complexity attacks} against hash
tables~\cite{CrosbyWa03,CaiGuJo09}, peacock and cuckoo hash
tables~\cite{PoratBrL12}, (unbalanced) binary search trees~\cite{CrosbyWa03},
and quicksort~\cite{KhanTr05}.  In these attacks, the adversary adaptively
constructs a sequence of inputs and queries that causes the data structure to
exhibit its worst-case performance.

\section*{{\bf A{\small CKNOWLEDGMENTS}}}

This research was supported in part by NSF grants
{\acksmall CCF 1114809,
CCF 1217708,
CCF 1218188,
CCF 1314633,
 CCF 1637458, 
IIS 1247726, 
IIS 1251137, 
CNS 1408695, 
CNS 1408782, 
CCF 1439084, 
CCF-BSF 1716252, 
CCF 1617618, 
IIS 1541613}, 
CAREER Award {\acksmall CCF 1553385}, 
as well as NIH grant {\acksmall 1U01CA198952-01},  
by the European Research Council under the European Union's 7th Framework Programme {\acksmall (FP7/2007-2013)}~/~ERC grant agreement no. {\acksmall 614331},
by Sandia National Laboratories, EMC, Inc, by  NetAPP, Inc,
  and the 
  VILLUM Foundation grant {\acksmall 16582}.

We thank Rasmus Pagh for helpful discussions.

\bibliographystyle{abbrv}

\appendix
\section{\bf U{\small SING} W{\small ORD}-L{\small EVEL} P{\small ARALLELISM}}\seclabel{wordstuff}

This  section explains that we can store and maintain small
strings (fingerprints or metadata) compactly within words, while
retaining $O(1)$ lookup (e.g., prefix match), insert, and delete.
Based on this section, we can just focus on where (i.e., in which word)
data is stored, and we use lemma~\lemref{parallelsearch} as a black box and 
assume that we can do all manipulations within machine words in $O(1)$~time. (For simplicity, we assume that machine words have 
$\Theta(\log n)$ bits, since words cannot be shorter and it does not hurt if they are longer.)

\begin{lemma} 
  \lemlabel{parallelsearch} Consider the following input.
  \begin{itemize}
    \setlength{\itemindent}{-1em}
  \item A ``query'' bit string $q$ that fits within a $O(1)$  $O(\log n)$-bit machine
    words. 
  \item ``Target'' strings~$s_1, s_2, \ldots , s_\ell$ concatenated
    together so that~$s=s_1\circ s_2 \circ \cdots \circ s_\ell$ fits
    within $O(1)$  $O(\log n)$-bit machine words. The strings need not be 
    sorted and need not have the same length.
  \item Metadata bits, e.g., a bit mask that indicates the starting
    bit location for each $s_i$.
  \end{itemize}
  Then the following query and update
  operations take~$O(1)$ time. 
  (Query results are returned as bit maps.)
  \begin{enumerate}
    \setlength{\itemindent}{-1em}
  \item \emph{Prefix match query.} Given a range $[j,k]$, find all
    $s_i$ such that $i\in[j,k]$ and $s_i$ is a prefix of $q$. 
  \item \emph{Prefix-length query.} More generally, given a range
    $[j,k]$, for each $s_i$ such that $i\in[j,k]$, indicate the length
    of the prefix match between $q$ and $s_i$. 
  \item \emph{Insert.} Given an input string $\hat{s}$ and target
    rank $i$, insert $\hat{s}$ into $s$ as the new  string of rank
    $i$. 
  \item \emph{Delete.} Given a target rank $i$, delete
    $s_i$ from $s$.
   \item \emph{Splice, concatenate, and other string operations.} For example, given $i$, concatenate 
$s_{i}$ with $s_{i+1}$. Given $i$ and length $x$, remove up to $x$ bits from the beginning of $s_{i}$.
  \item \sloppy \emph{Parallel inserts, deletes, splices, concatenations, and queries.} Multiple insert, delete, query, splice, or concatenate operations can be supported in parallel. For example, remove $x$ bits from the beginning of all strings, and indicate which strings still have bits.
  \end{enumerate}
\end{lemma}

\begin{proof}
  These operations can be supported using standard compact and
  succinct data-structures techniques, in particular, rank, select,
  mask, shift, as well as sublinear-sized lookup tables.

  Insertion, deletion, concatenation, splicing, etc., are handled
  using simple selects, bit-shifts, masks, etc.

  We first explain how to find a string $s_i$ that is a prefix match
  of $q$ for the special case that $q$ has length at most $\log n/8$.
  We partition $s$ into chunks such that each chunk has size
  $\log n/8$. Now some $s_j$ are entirely contained within one chunk
  and some straddle a chunk boundary. Since only $O(1)$ strings can
  straddle a chunk boundary, we can search each of these strings
  serially. 
  
  In contrast, there may be many strings that are entirely contained
  within a chunk, and these we need to search in parallel.  We can use
  lookup tables for these parallel searches. Since there are at most
  $2^{\log n/8}=n^{1/8}$ input choices for query string $q$ and at
  most $n^{1/4}$ input choices for the concatenated target strings
  (metadata bits and $s$), a lookup table with precomputed responses
  to all possible queries still takes $o(n)$ bits.

 The remaining operations are supported similarly using rank, select, and lookup tables. 
 For example, the more general case of querying
  larger $q$ is also supported similarly
  by dividing $q$  into chunks, comparing the chunks iteratively, and doing further parallel
  manipulation on $s$, also using lookup tables. In particular, since we compare $q$ in chunks, we have to remove all strings $s_i$ that already do not have a prefix in common with $q$ as well as those shorter strings where no prefix is left. 
\end{proof}

\end{document}